\documentclass{cas-sc}
\usepackage{amsthm}
\usepackage{tabularx}
\usepackage{longtable}
\usepackage{needspace}
\usepackage{bm}
\usepackage[section]{placeins}  %
\usepackage{algorithm}
\usepackage{algpseudocode}
\usepackage{mathtools}

\usepackage{natbib}
\hypersetup{hypertexnames=false}
\ExplSyntaxOn
\keys_set:nn { stm / mktitle } { nologo = true }
\ExplSyntaxOff

\newcommand{\diag}{\operatorname{diag}}

\makeatletter
\let\cas@hypertarget@orig\hypertarget
\renewcommand{\hypertarget}[2]{%
  \def\cas@tmp{#1}%
  \ifx\cas@tmp\empty
    #2%
  \else
    \cas@hypertarget@orig{#1}{#2}%
  \fi
}
\makeatother

\usepackage{bibunits}
\makeatletter
\def\@bibunitname{\jobname-bu\the\@bibunitauxcnt}
\makeatother
\usepackage{xr-hyper}
\usepackage{tikz}
\usetikzlibrary{arrows.meta,positioning,fit,decorations.pathreplacing,calc}

\definecolor{cbBlue}{HTML}{0072B2}
\definecolor{cbOrange}{HTML}{E69F00}
\definecolor{cbTeal}{HTML}{009E73}
\definecolor{cbPurple}{HTML}{CC79A7}
\definecolor{cbSky}{HTML}{56B4E9}
\definecolor{cbVermillion}{HTML}{D55E00}

\tikzset{
  process/.style = {rectangle, draw, rounded corners, minimum width=2.4cm,
    minimum height=0.7cm, align=center, font=\small},
  arrow/.style = {-{Stealth[length=5pt]}, line width=0.95pt, draw=black!70},
  dasharrow/.style = {arrow, dashed},
  figpanel/.style = {font=\footnotesize\bfseries},
  note/.style = {font=\scriptsize, text=black!75, align=center},
  subtleframe/.style = {draw=black!45, rounded corners, line width=0.7pt},
  boxneutral/.style = {draw=black!60, rounded corners, fill=gray!8, line width=0.8pt},
  boxblue/.style = {draw=cbBlue!85!black, rounded corners, fill=cbBlue!10, line width=0.9pt},
  boxorange/.style = {draw=cbOrange!90!black, rounded corners, fill=cbOrange!14, line width=0.9pt},
  boxteal/.style = {draw=cbTeal!80!black, rounded corners, fill=cbTeal!12, line width=0.9pt},
  boxpurple/.style = {draw=cbPurple!80!black, rounded corners, fill=cbPurple!12, line width=0.9pt},
  labelblue/.style = {text=cbBlue!80!black, font=\scriptsize\bfseries},
  labelorange/.style = {text=cbOrange!90!black, font=\scriptsize\bfseries},
  labelteal/.style = {text=cbTeal!80!black, font=\scriptsize\bfseries},
  labelpurple/.style = {text=cbPurple!80!black, font=\scriptsize\bfseries},
} 
\newif\ifreview\reviewfalse
\newcommand{\crossdocnumber}[2]{\expandafter\gdef\csname crossdoc@#1\endcsname{#2}}
\crossdocnumber{alg:final_procedure}{2}
\crossdocnumber{eq:cht_backward_euler}{2}
\crossdocnumber{eq:cht_kkt_acceptance}{25}
\crossdocnumber{eq:cht_temporal_eigenproblem}{17}
\crossdocnumber{eq:cht_trajectory_operator}{3}
\crossdocnumber{eq:restriction_identity}{20}
\crossdocnumber{eq:revision_projectors}{18}
\crossdocnumber{eq:revision_reduced}{5}
\crossdocnumber{eq:skew_transport}{10}
\crossdocnumber{fig:mesh_sequences}{5}
\crossdocnumber{fig:optimized_temperature_bounds}{4}
\crossdocnumber{sec:v20_results}{3.4}
\crossdocnumber{si:cht_timing}{S2.3.3}
\crossdocnumber{si:cht_validation_results}{S3.1.2}
\crossdocnumber{si:limits}{S3.5}
\crossdocnumber{si:mesh_formulation}{S2.1.1}
\crossdocnumber{si:mesh_reference}{S2.2.1}
\crossdocnumber{si:mesh_refinement_diagnosis}{S3.5.1}
\crossdocnumber{si:mesh_results}{S3.1.1}
\crossdocnumber{si:mesh_targets}{S2.1.3}
\crossdocnumber{si:methodology}{S2}
\crossdocnumber{si:tab_amgx_configuration}{S3}
\crossdocnumber{si:tab_cht_conditioning}{S13}
\crossdocnumber{si:tab_experiment_map}{S2}
\crossdocnumber{si:tab_mesh_transfer}{S10}
\crossdocnumber{si:tab_retention_ablation}{S12}
\crossdocnumber{si:tab_small_spectrum}{S11}
\crossdocnumber{si:tab_v20_comparisons}{S14}
\crossdocnumber{si:tab_v20_direct}{S16}
\crossdocnumber{si:tab_v20_single}{S15}
\crossdocnumber{si:temperature_bounds}{S3.4.2}
\crossdocnumber{si:temporal_coverage}{S2.2.3}
\crossdocnumber{si:theory}{S2.2.4}
\crossdocnumber{si:transient}{S2.1.4}
\crossdocnumber{si:transport_forms}{S2.1.2}
\crossdocnumber{si:v20_temporal}{S3.1.3}
\crossdocnumber{tab:cht_transfer}{3}
\crossdocnumber{tab:mesh_complete}{2}
\crossdocnumber{tab:v20_operation}{6}
\crossdocnumber{tab:v20_scale}{4}
\crossdocnumber{tab:v20_transient}{5} 
\makeatletter
\DeclareRobustCommand{\crossdocref}[1]{%
  \ifcsname crossdoc@#1\endcsname\csname crossdoc@#1\endcsname
  \else\textbf{\char63\char63}\@latex@warning{Reference `#1' on page \thepage\space undefined}\fi}
\makeatother
\newcommand{\Sref}[1]{\ifreview\ref{#1}\else\crossdocref{#1}\fi}


\newcommand{\inputgated}[2]{\IfFileExists{#1}{\input{#2}}{}}
\ExplSyntaxOn
\tl_new:N \l_solverrows_file_tl
\seq_new:N \l_solverrows_seq
\int_new:N \l_solverrows_count_int
\bool_new:N \l_solverrows_break_bool
\NewDocumentCommand{\solverrowprefix}{mmm}
  {\tl_gset:cn {g_solverrowprefix_#1_#2_tl} {#3}}
\NewDocumentCommand{\loadsolverrows}{O{}mm}
  {
    \file_get:nnN {#2} {} \l_solverrows_file_tl
    \seq_set_split:NnV \l_solverrows_seq {\\} \l_solverrows_file_tl
    \tl_gclear_new:c {g_solverrows_#3_tl}
    \int_zero:N \l_solverrows_count_int
    \bool_set_false:N \l_solverrows_break_bool
    \seq_map_inline:Nn \l_solverrows_seq
      {
        \tl_if_blank:nF {##1}
          {
            \cs_if_exist:cT
              {g_solverrowprefix_#3_\int_eval:n {\l_solverrows_count_int+1}_tl}
              {
                \tl_gput_right:cx {g_solverrows_#3_tl}
                  {\exp_not:v {g_solverrowprefix_#3_\int_eval:n {\l_solverrows_count_int+1}_tl}}
              }
            \tl_gput_right:cn {g_solverrows_#3_tl} {##1}
            \int_incr:N \l_solverrows_count_int
            \tl_if_blank:nTF {#1}
              {
                \int_compare:nNnTF {\int_mod:nn {\l_solverrows_count_int} {4}} = {0}
                  {\bool_set_true:N \l_solverrows_break_bool}
                  {\bool_set_false:N \l_solverrows_break_bool}
              }
              {
                \exp_args:Nnx \clist_if_in:nnTF {#1} {\int_use:N \l_solverrows_count_int}
                  {\bool_set_true:N \l_solverrows_break_bool}
                  {\bool_set_false:N \l_solverrows_break_bool}
              }
            \bool_if:NTF \l_solverrows_break_bool
              {\tl_gput_right:cn {g_solverrows_#3_tl} {\\}}
              {\tl_gput_right:cn {g_solverrows_#3_tl} {\\*}}
          }
      }
    \bool_if:NF \l_solverrows_break_bool
      {\PackageError{solverrows}{Incomplete~solver~group~in~#2}{Check~the~source~rows.}}
  }
\NewExpandableDocumentCommand{\usesolverrows}{m}{\tl_use:c {g_solverrows_#1_tl}}
\ExplSyntaxOff
\makeatletter
\newcommand{\bibnamespace}[1]{\def\@extra@binfo{.#1}\def\@extra@b@citeb{.#1}}
\makeatother
\newcommand{\sisummary}{\par\medskip\noindent\textbf{Supporting information.}\enspace} 
\begin{document}
\begin{bibunit}[cas-model2-names]
\bibnamespace{main}
\def\lastpage{\pageref*{main:lastpage}}
\shorttitle{Reference-Space Reuse for Constrained Heat Transfer}
\shortauthors{Kadeethum et al.}

\title[mode = title]{Reference-Space Reuse for Steady and Transient State-Constrained Conjugate Heat Transfer}

\author[siemens]{Teeratorn Kadeethum}
\cormark[1]
\ead{tkadeethum@gmail.com}

\author[ucsc1,ucsc2]{Francesco Ballarin}[orcid=0000-0001-6460-3538]
\ead{francesco.ballarin@unicatt.it}

\author[llnl]{Youngsoo Choi}
\ead{choi15@llnl.gov}

\author[fsu]{Sanghyun Lee}
\ead{slee17@fsu.edu}

\affiliation[siemens]{organization={Siemens Energy, Artificial Intelligence (AI) Lab},
  city={Orlando},
  state={FL},
  country={United States}}

\affiliation[ucsc1]{
  organization={Department of Mathematics and Physics, Università Cattolica del Sacro Cuore},
  addressline={via Garzetta 48},
  city={Brescia},
  postcode={25133},
  country={Italy}
}

\affiliation[ucsc2]{
  organization={Department of Mathematics for Economic, Financial and Actuarial Sciences, Università Cattolica del Sacro Cuore},
  addressline={via Necchi 9},
  city={Milano},
  postcode={20123},
  country={Italy}
}

\affiliation[llnl]{organization={Lawrence Livermore National Laboratory},
  city={Livermore},
  state={CA},
  country={United States}}

\affiliation[fsu]{organization={Florida State University},
  city={Tallahassee},
  state={FL},
  country={United States}}

\cortext[1]{Corresponding author}

\begin{abstract}
Repeated thermal optimization must enforce local temperature limits
while solving linear systems whose unknowns change as the constraints
activate and release. Zero-extension transfer, which carries correction vectors
from one inactive set to the next, assigns zero to the degrees of freedom that a
constraint has just released. We construct a spectral reference on
the full state domain before optimization and restrict it directly to the
current inactive set at each update, so the coarse correction receives
reference entries at every released degree of freedom. Spatial and
space--time constructions cover steady and transient conjugate heat
transfer with prescribed fluid velocity, and a restriction-residual identity
yields alignment and deflated-convergence bounds under stated rank and
spectral-separation assumptions. In the primary comparisons on body-fitted transformer
and bore-in-block models, the resulting reference deflation reduces the median complete
optimization time relative to the tested iterative alternatives on the same graphics
processor. Complete time includes reference construction. Matched transfer experiments
connect the improvement to the reference information retained at released constraints.
The lower cost persists on refined Cartesian and bore-in-block meshes, on space--time
trajectories and over sequences
with changing loads and query order; the margin narrows, and eventually reverses, as a
trajectory gains time slabs at a fixed horizon. Full-domain
reference retention thus supports repeated evaluation of
temperature-constrained heat inputs in engineering thermal models.

\end{abstract}


\maketitle
\hypersetup{pdfauthor={Teeratorn Kadeethum, Francesco Ballarin, Youngsoo Choi, Sanghyun Lee},
  pdfsubject={State-constrained conjugate heat transfer optimization}} 
\section{Introduction}
\label{sec:introduction}

Local overheating limits the operation and service life of electrical
equipment. In oil-cooled transformers, winding hot spots accelerate
insulation ageing, which motivates explicit limits on temperature
\citep{li2024thermofluidic}.
Thermal control seeks a heat input that brings a temperature field close
to a desired distribution and respects local temperature limits.
A small average tracking error can coexist with a localized temperature
violation. Pointwise state constraints address this distinction by
including the temperature bounds in the optimization problem
\citep{hinze2009optimization,troeltzsch2010optimal,leugering2012constrained}.
We consider conjugate heat transfer (CHT), where conduction through a
solid couples to heat transport in a fluid. The control is a distributed
heat source, and the velocity field and material properties are prescribed.

Repeated optimization arises when a study considers several desired
temperature distributions or trajectories
\citep{benner2017model,kouri2018optimization}. Each query introduces a
second level of repetition: the optimizer updates the constrained
temperature variables and solves the resulting linear equations several
times. Transient optimization couples time levels through the state
equations and tracks the desired temperature throughout the trajectory.
Each transient query therefore computes a complete controlled trajectory.
Reusing information across these linear solves
offers a way to reduce the cost of the optimization sequence.

A primal-dual active-set (PDAS) method uses the temperature and multiplier
values to update the active set \citep{hintermueller2003primal}.
During each step, an active degree of freedom (DOF) is fixed at its
temperature bound; inactive DOFs remain unknown. For a linear
state equation and a quadratic objective, eliminating the distributed
control gives a symmetric positive-definite (SPD) reduced matrix. Inactive
equations use a principal submatrix of this matrix. Even when
the geometry, material properties, and prescribed flow remain unchanged,
enforcing temperature limits produces a sequence of different inactive
systems. Their dimensions and entries change as the constraints activate
and release, which complicates the reuse of solver information.

Structure-aware preconditioners reduce the work of solving constrained
optimality systems. Full-space and block methods exploit the coupled
equations \citep{choi2012simultaneous,rees2010optimal,pearson2012new}.
Active-set-aware preconditioners account for changing constraints
\citep{porcelli2015preconditioning}. Schur-complement factorizations
provide related strategies \citep{choi2015practical}.
Their setup and update costs form part of the repeated optimization
work that a reuse policy must address.

Reduced representations support repeated calculations by compressing
the state and optimality equations
\citep{quarteroni2016reduced,hesthaven2016certified,negri2013reduced}.
Applications include sensitivity evaluation
\citep{amsallem2015design,choi2020gradient}, component-wise approximation
\citep{mcbane2021component,mcbane2022stress,hoang2021domain} and
joint space--time reduction \citep{choi2019space}.
Their use within full linear solves \citep{choi2019accelerating}
motivates assessing reusable correction spaces through the complete
cost of solving the original discrete equations.

Deflation and recycling retain correction directions for a full-space
Krylov solve. Spectral deflation removes selected error components
through a coarse correction
\citep{nicolaides1987deflation,frank2001construction,saad2000deflated};
its relationships with augmentation and two-level preconditioning are
well established \citep{tang2009comparison,gaul2014framework}.
Recycling updates a subspace from preceding solves
\citep{parks2006recycling,soodhalter2014krylov,soodhalter2020survey}.
Transfer rules carry recycled vectors across changing dimensions: for example,
\citet{bolten2022evolving} map recycle vectors between evolving
finite-element meshes. Reusable spaces can also be constructed offline.
\citet{panagiotopoulos2023accelerated} construct a recycle space over a
parameter interval and use reduced representations to accelerate its
online projector construction for affine parametric systems.
These precedents motivate a matched comparison of specified transfer
policies when temperature constraints change.

In state-constrained optimization, a correction vector stored only on the
preceding inactive set needs a transfer rule at each newly inactive node. Zero-extension
transfer, the simplest rule, assigns zero to newly released degrees of freedom; a
recycling solver can afterwards learn directions there from its own inner solves. A
reference retained on the full state domain already has values at those nodes.
We investigate whether this information improves the coarse correction enough
to reduce the complete cost of repeated optimization, including the construction
and application of the reference.

We construct a spectral reference on the full state domain before optimization
and restrict it directly to the current inactive set at each update; earlier
restrictions leave the stored reference unchanged. The primal-dual active-set
method, the coarse correction, deflation and Krylov recycling are established
components. The contribution is the retention policy with direct restriction,
supported by spatial and space--time reference constructions, a conditional
analysis and a complete-cost evaluation. The evaluation compares complete
constrained optimization with the resulting solver, reference deflation (RefDef), against
Jacobi-preconditioned conjugate gradient (CG), abbreviated Jacobi-CG; recycling
with retained coarse vectors and Jacobi-scaled selection; and the algebraic multigrid
(AMG) solver of AmgX \citep{naumov2015amgx}.

Section~\ref{sec:methodology} presents the formulation, the reference policy, its
analysis and the timing procedures. Section~\ref{sec:results} reports the
comparisons, Section~\ref{sec:discussion} interprets them and states their scope,
and Section~\ref{sec:conclusions} concludes. 
\section{Methodology}
\label{sec:methodology}
The optimization couples solid and fluid temperatures through their
material interfaces and retains temporal coupling across each transient
trajectory. Fluid velocity and material properties remain prescribed.

\subsection{Steady and transient constrained heat transfer}
\label{sec:formulation}

The temperature field in the solid and fluid satisfies
\begin{equation}\label{eq:cht_transient_pde}
 c(x)\,\partial_t y-\nabla\cdot(\kappa(x)\nabla y)
       +v(x)\cdot\nabla y=u+f_s.
\end{equation}
The transport coefficient is $v=c\beta$, where the prescribed fluid
velocity $\beta$ vanishes in the solid. Conductivity is positive in
isotropic materials and a positive-definite tensor in the anisotropic
winding. Thermal capacity is positive. These properties remain fixed
during optimization. The distributed control $u$ changes heat input throughout
both materials and permits heating and cooling without magnitude limits.
Steady optimization omits the time derivative.
Temperature is continuous at a material interface; the conductive-flux
jump equals a prescribed interface source when one is present.
Dirichlet temperatures and the initial temperature are prescribed.
The upper bound $y\leq\psi$ applies at every optimized state DOF and every positive time level.

The primary application models are an axisymmetric transformer with
oil channels, an anisotropic winding and insulating baffles, and a
three-dimensional bore inside a conducting block.
Figure~\ref{fig:mesh_showcases} shows their geometry and changing
temperature constraints. The transformer uses a fixed prescribed
velocity through the oil channels. The bore-in-block model uses
prescribed azimuthal transport in the gas, an interface heat source and
a fixed coolant temperature on the outer block. Bore 1 to Bore 4 denote the
bore-in-block meshes obtained by refining the packaged coarse mesh one to four
times; the tables give their dimensions.

An optimization query $\theta_m$ specifies a desired temperature
$d(\theta_m)$ or trajectory $d(t,\theta_m)$. Each primary sequence changes
only this target, at a fixed heterogeneous operator, regularization and
physical bound; Section~\ref{sec:v20_results} also varies loads, flow rates, chained
initial states and query order. The index $m$ labels complete optimization problems; $n$ labels
physical time levels within one transient problem.

\begin{figure}[pos=htbp]
\centering
\includegraphics[width=\textwidth]{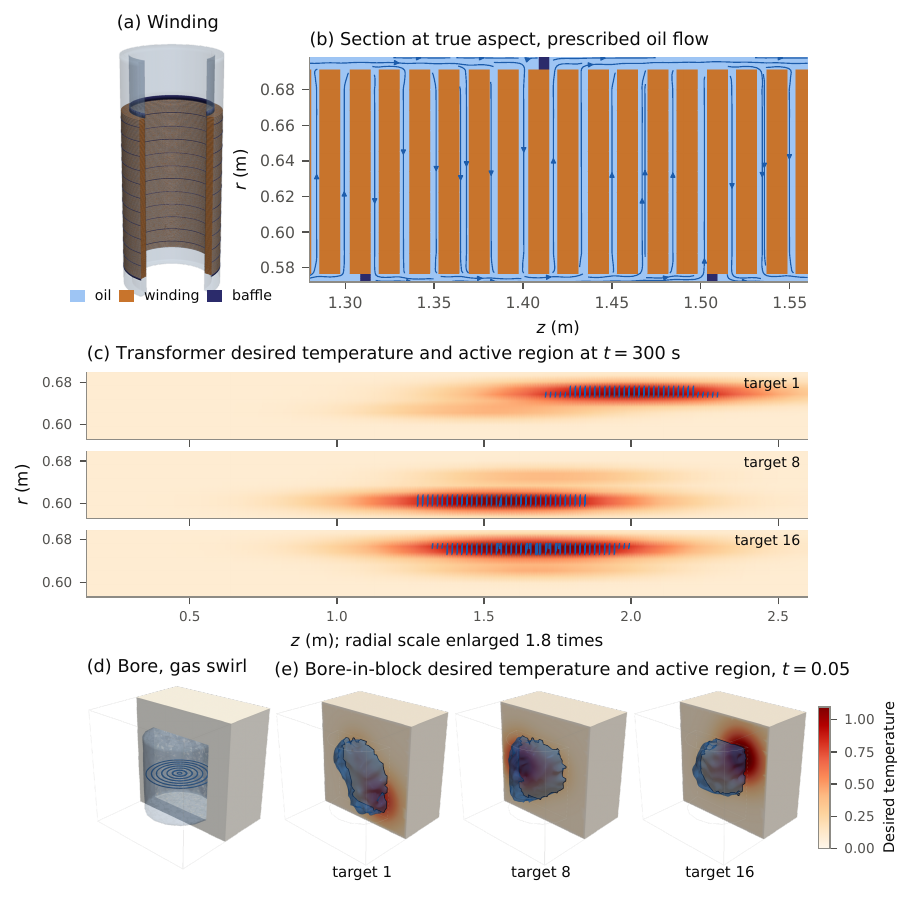}
\caption{Body-fitted application models and their changing temperature
constraints. (a) Disc-type transformer winding revolved from its axisymmetric
section: translucent oil, copper discs and baffles. (b) Part of the section at
true aspect with streamlines of the prescribed oil velocity, which pass through
the radial ducts between discs and around the baffles. (c) Desired temperature
and active region (blue) for targets 1, 8 and 16 at the second of four time
levels, 300~s. (d) Bore-in-block model (Bore~2) with the front half of the block removed:
the gas bore (translucent) and streamlines of the prescribed azimuthal swirl.
(e) Desired temperature on the exposed section $x_2/L=0.5$ and active region
(blue) for the same targets at dimensionless time 0.05. The active regions come from
the recorded four-slab sequences of Table~\ref{tab:mesh_complete}, the transformer
with the skew-symmetric transport. Temperatures are normalized by the model-specific
scales of Section~\ref{sec:mesh_formulation}, and panels (c) and (e) share one color
scale.}
\label{fig:mesh_showcases}
\end{figure} 

\subsubsection{Weighted discrete optimization}
Let $A_s$ be the discrete spatial state operator after Dirichlet
elimination, $C$ its positive diagonal capacity matrix, and $N_s$
the number of optimized spatial DOFs. Positive diagonal
weights $W_s$ define both state tracking and the distributed-control
penalty. These quantities cover the finite-element and finite-difference
discretizations described below.

For $0=t_0<t_1<\cdots<t_{N_t}=T_f$, set
$\Delta t_n=t_n-t_{n-1}$; each backward Euler step defines one time slab.
Backward Euler gives
\begin{equation}\label{eq:cht_backward_euler}
 (A_s+C/\Delta t_n)y^n-(C/\Delta t_n)y^{n-1}=u^n+f_s.
\end{equation}
The vectors $Y=(y^1,\ldots,y^{N_t})$ and
$U=(u^1,\ldots,u^{N_t})$ are ordered by time level, then by spatial
node. Define the lower bidiagonal temporal matrix
$L_{nn}=1/\Delta t_n$ and $L_{n,n-1}=-1/\Delta t_n$.
The coupled state equation is
\begin{equation}\label{eq:cht_trajectory_operator}
 \mathcal A Y=U+f,\qquad
 \mathcal A=I_{N_t}\otimes A_s+L\otimes C,\qquad
 f=\boldsymbol 1\otimes f_s+(Cy^0/\Delta t_1,0,\ldots,0).
\end{equation}
The space--time operator $\mathcal A$ couples each time block to the preceding one,
and its transpose $\mathcal A^\top$, which enters the reduced Hessian and the adjoint
equation~\eqref{eq:control_adjoint}, couples each adjoint block to the following time
block. Both actions retain their temporal coupling throughout optimization.

Let $\overline{\Delta t}=T_f/N_t$ and
$W_t=\diag(\Delta t_n/\overline{\Delta t})$.
Right-endpoint quadrature defines $W=W_t\otimes W_s$.
The transient discrete objective is
\begin{equation}\label{eq:cht_objective}
 \frac{\gamma}{2}
 \left((Y-d)^\top W(Y-d)+\alpha U^\top WU\right),
 \qquad \alpha>0,\quad\gamma>0.
\end{equation}
The common factor $\gamma$ contains the spatial and temporal integration
scales. Both terms use the same positive weights. Eliminating
$U=\mathcal A Y-f$ and dividing by the positive common factor
$\gamma$ gives
\begin{equation}\label{eq:revision_reduced}
 \min_{Y\leq\psi}\ \frac12Y^\top MY-\ell^\top Y,
 \qquad
 M=W+\alpha\mathcal A^\top W\mathcal A,\qquad
 \ell=Wd+\alpha\mathcal A^\top Wf.
\end{equation}
Here $\psi$ also denotes the vector with the prescribed scalar bound
at each DOF. The steady problem has the same reduction
with $\mathcal A=A_s$, $W=W_s$ and $f=f_s$.
The positive tracking weights make $M$ SPD. The normal operator
retains the $L^2$ control penalty and its associated conditioning.

For an upper-bound multiplier $\lambda$, the discrete
Karush--Kuhn--Tucker (KKT) conditions are
\begin{equation}\label{eq:revision_kkt}
 MY-\ell+\lambda=0,\qquad
 Y\leq\psi,\quad\lambda\geq0,\quad
 \lambda_i(\psi_i-Y_i)=0.
\end{equation}
Strict convexity makes these conditions necessary and sufficient.
At a PDAS step, let $J$ be the active set and $I$ its inactive
complement. The restriction matrix $R_I$ selects the inactive
DOFs, and $R_I^\top$ extends a vector by zero.
Imposing $Y_J=\psi_J$ and $\lambda_I=0$ gives
\begin{samepage}
\begin{equation}\label{eq:revision_inactive}
 B Y_I=b,\qquad B=M_{II},\qquad b=\ell_I-M_{IJ}\psi_J.
\end{equation}
The full reduced operator supplies both blocks. Its temporal
off-diagonal blocks contribute to this system whenever the relevant
time-level DOFs remain inactive.
\end{samepage}

\Needspace{10\baselineskip}
\subsubsection{Active-set update and recovery}
After the inactive solve, PDAS recovers
$\lambda_J=(\ell-MY)_J$ and updates
\begin{equation}\label{eq:complementarity_map}
 J_{\rm new}=\{i:\lambda_i+M_{ii}(Y_i-\psi_i)>0\}.
\end{equation}
An active DOF can therefore be released when its
multiplier becomes negative. The complete solver checks stationarity,
primal feasibility, dual feasibility, complementarity and the
projected gradient. Section~\ref{sec:benchmarks} gives their
normalization and acceptance thresholds.

The normalized discrete Lagrangian uses the unweighted state
equation $\mathcal A Y-U-f=0$. Its adjoint multiplier $p$ and the
control satisfy
\begin{equation}\label{eq:control_adjoint}
 U=\mathcal A Y-f,\qquad p=\alpha WU,\qquad
 W(Y-d)+\mathcal A^\top p+\lambda=0.
\end{equation}
The physical objective and its multipliers recover the common factor
removed in \eqref{eq:revision_reduced}. Control recovery uses all
state values, including the prescribed active values. The transient
verification also evaluates independent forward and transpose actions
of \eqref{eq:cht_backward_euler}, with nonzero initial data and
unequal time steps.

\subsubsection{Body-fitted thermal discretizations}
\label{sec:mesh_formulation}
Conforming piecewise-linear (P1) finite elements discretize temperature
on triangles or tetrahedra.
The transformer integrates with the axisymmetric measure $2\pi r\,dr\,dz$.
Let $M_\ell$ and $C_\ell$ denote the positive lumped mass and capacity
matrices, and let $K$ include diffusion, prescribed transport, and the
specified streamline diffusion. On the transformer, the transport takes the
skew-symmetric form
\begin{equation}\label{eq:skew_transport}
 \int_\Omega 2\pi r\,c\,\zeta\,(\beta_h\cdot\nabla y)
 +\tfrac12\int_\Omega 2\pi r\,c\,(\nabla\cdot\beta_h)\,y\,\zeta
\end{equation}
for temperature $y$, test function $\zeta$ and volumetric capacity $c$, with the
prescribed P2 Stokes velocity $\beta_h$, which is divergence-free only weakly. The skew
form contributes the outflow term to the discrete energy estimate. The symmetric part of
the assembled free-node operator is positive definite for the tested transformer and
bore-in-block trajectory models, which rules out growing backward Euler modes.

After lifting and eliminating Dirichlet
values, the state equation gives
$A_s=M_\ell^{-1}K$, capacity $C=M_\ell^{-1}C_\ell$, and a fixed load $f_s$.
The trajectory load in \eqref{eq:cht_trajectory_operator} includes $f_s$
at every time level and the initial-temperature contribution at the first.
The tracking and control weights become
$W=W_t\otimes W_s$, where
$W_s=\diag(m_i/\bar m)$, $m_i=(M_\ell)_{ii}$, and $\bar m$ is their
mean over the optimized spatial DOFs. The steady problem
uses $W_s$, and its physical integration factor is $\bar m$.
The transient factor is $\gamma=\bar m\overline{\Delta t}$. These positive weights preserve the reduced equations,
PDAS update, and recovery relations above.

The transformer uses a temperature scale of 20~K above its 341.3~K inlet,
bound $\psi=0.8$, and $\alpha=10^{-14}$. Its transient horizon is 600~s.
The bore-in-block model uses conductivity ratio 1000 and unit dimensionless thermal
capacity in both materials. Its cost comparisons use bound $\psi=0.2$
and $\alpha=10^{-6}$.
Its dimensionless transient horizon is 0.1. Each query changes the desired
temperature profile and retains these physical coefficients and bounds.
Nested subdivision preserves the material partition and boundary labels.

\subsubsection{Cartesian construction and refinement comparisons}
The unit cube $\Omega=(0,1)^3$ contains a solid region
$\Omega_s=\{x_3<1/2\}$ and a fluid region $\Omega_f=\{x_3\geq1/2\}$.
The conductivity is 100 in the solid and 1 in the fluid; capacity
is 2 and 1, respectively. The prescribed transport is
\begin{equation}\label{eq:cht_transport}
 v\cdot\nabla y=
 \begin{cases}
 0,&x_3<1/2,\\
 200x_3(1-x_3)\,\partial_{x_1}y,&x_3\geq1/2.
 \end{cases}
\end{equation}
The boundary and initial temperatures are zero, and $f_s=0$.
The trajectories use horizon $T_f=0.1$.

A grid with $n_s$ interior nodes per coordinate has spacing
$h=(n_s+1)^{-1}$ and $N_s=n_s^3$ DOFs.
Harmonic face conductivities and centered transport differences
define $A_s$; the differences are skew-symmetric because the velocity varies only
across the flow direction. These comparisons use $W_s=I$, $\alpha=10^{-3}$ and
objective factor $h^3$ for steady problems or
$h^3\overline{\Delta t}$ for trajectories. This specialization
provides analytical reference modes and permits separate spatial,
temporal and rank studies under the same optimization formulation.
\FloatBarrier 
\FloatBarrier
\subsection{Full-domain reference construction}
\label{sec:reference}

The reference policy fixes correction information before the sequence
of optimization queries. Let $N$ be the number of full-domain state
DOFs: $N=N_s$ for steady problems and $N=N_sN_t$ for
transient problems. A reference matrix $M_0$ supplies $r$ columns
$\Phi\in\mathbb R^{N\times r}$. The implementation stores either these
columns or factors that evaluate them at any state DOF.
For each current inactive set, it forms
\begin{equation}\label{eq:reference_restriction}
 \widetilde Z=R_I\Phi.
\end{equation}
The reusable object is the full-domain reference space; the restricted
basis and its coarse factorization are reconstructed for each inactive
system. The full reference remains unchanged after restriction. Thus the
next inactive set receives values from the same full-domain space,
including entries at newly released constraints.
\subsubsection{Reference information at released constraints}
Let $R_a$ and $R_b$ restrict to consecutive inactive sets.
Direct full-reference restriction gives $R_b\Phi$.
Zero-extension transfer of the preceding restriction gives
$R_bR_a^\top R_a\Phi$. Their difference is
\begin{equation}\label{eq:full_reference_transfer}
 R_b\Phi-R_bR_a^\top R_a\Phi
 =R_b(I-R_a^\top R_a)\Phi.
\end{equation}
The right-hand side contains the reference values at newly inactive
DOFs. It vanishes at DOFs shared by the
two inactive sets. Repeated zero-extension transfer can additionally
accumulate losses from earlier restrictions.
Figure~\ref{fig:reference_retention} illustrates the supplied entries
for one reference column.
\begin{figure}[pos=htbp]
\centering
\begin{tikzpicture}[x=1.05cm,y=0.67cm,font=\small]
\node[anchor=east] at (0,0) {Full reference};
\node[anchor=east] at (0,-1) {Previous inactive set};
\node[anchor=east] at (0,-2) {Direct restriction};
\node[anchor=east] at (0,-3) {Zero-extension transfer};
\foreach \j in {1,...,6} {
  \node[draw,minimum width=0.95cm,minimum height=0.5cm] at (\j,0) {$\phi(x_{\j})$};
}
\foreach \j in {1,2,4,6} {
  \node[draw,fill=blue!8,minimum width=0.95cm,minimum height=0.5cm] at (\j,-1) {$\phi(x_{\j})$};
}
\foreach \j in {3,5} {
  \node[draw,fill=black!15,minimum width=0.95cm,minimum height=0.5cm] at (\j,-1) {active};
}
\foreach \j in {1,2,6} {
  \node[draw,fill=blue!8,minimum width=0.95cm,minimum height=0.5cm] at (\j,-2) {$\phi(x_{\j})$};
  \node[draw,fill=blue!8,minimum width=0.95cm,minimum height=0.5cm] at (\j,-3) {$\phi(x_{\j})$};
}
\foreach \j in {3,5} {
  \node[draw,fill=orange!25,minimum width=0.95cm,minimum height=0.5cm] at (\j,-2) {$\phi(x_{\j})$};
  \node[draw,fill=orange!25,minimum width=0.95cm,minimum height=0.5cm] at (\j,-3) {$0$};
}
\foreach \j in {2,3} {
  \node[draw,fill=black!15,minimum width=0.95cm,minimum height=0.5cm] at (4,-\j) {active};
}
\node[anchor=north] at (3.5,-3.7)
 {Blue: shared inactive\quad Orange: newly released\quad Gray: active};
\end{tikzpicture}
\caption{Full-domain retention supplies reference entries at newly
released DOFs. The values $\phi(x_1),\ldots,\phi(x_6)$
belong to one reference column at six spatial nodes.
Direct restriction restores entries 3 and 5;
zero-extension transfer assigns zero at these entries. Both transfers
retain entries 1, 2 and 6 and remove newly active entry 4.}
\label{fig:reference_retention}
\end{figure} 

Equation~\eqref{eq:full_reference_transfer} identifies the information
provided by these two specified transfers. A recycling method can
estimate new entries through interpolation, approximation or
information generated by later Krylov iterations
\citep{bolten2022evolving}. The matched replay isolates full-reference
retention: it extracts the inactive systems of complete PDAS sequences and
applies both transfers, started from the same reference, to identical
matrices, right-hand sides and zero initial guesses. It records activation and release separately, together with
the actual numerical rank and the coarse correction's reduction of
the initial energy error.

Figure~\ref{fig:workflow} shows the reference lifetime and the
inner and outer solve loops.

\begin{figure}[pos=htbp]
\centering
\begin{tikzpicture}[>=Latex,
 block/.style={draw,rounded corners,align=center,font=\small,text width=3.2cm,minimum height=12mm}]
\node[block] (ref) at (0,0) {Once per sequence:\\construct spatial or\\space--time reference};
\node[block] (target) at (4,0) {Current target\\and active set};
\node[block] (restrict) at (8,0) {Each PDAS step:\\form inactive system,\\restrict basis and factor};
\node[block] (solve) at (8,-1.9) {Deflated CG\\recompute original residual};
\node[block] (kkt) at (4,-1.9) {Recover multipliers\\evaluate KKT conditions};
\node[block] (accepted) at (0,-1.9) {Accept state or trajectory,\\control and multipliers};
\draw[->] (ref.north) -- ++(0,5mm) -| (restrict.north);
\draw[->] (target)--(restrict);
\draw[->] (restrict)--(solve);
\draw[->] (solve)--(kkt);
\draw[->] (kkt)--node[right,font=\scriptsize]{update set}(target);
\draw[->] (kkt)--node[above,font=\scriptsize]{pass}(accepted);
\draw[->] (accepted.west) -- ++(-6mm,0) |- (-2.3,1.55)
 -- node[above,font=\scriptsize]{next target; reuse the same full-domain basis}
 (4,1.55) -- (target.north);
\end{tikzpicture}
\caption{The workflow reuses a full-domain reference space across
successive optimization problems. Reference construction occurs once
per sequence. Each PDAS step restricts
the same basis and solves the current inactive equations. Passing the KKT checks completes the optimization for the current
target. The next target begins with the selected warm-start policy.}
\label{fig:workflow}
\end{figure} 
\FloatBarrier

\subsubsection{Reference spaces on body-fitted meshes}
\label{sec:mesh_reference}
The body-fitted comparisons construct spatial information on one coarse
mesh. Let $M_c$ be its weighted steady reduced matrix and
$D_c=\diag(M_c)$. The generalized eigenproblem
$M_c\phi_j=\mu_jD_c\phi_j$ supplies low modes in Jacobi-scaled coordinates.
P1 prolongation maps these vectors to the full fine mesh, where each
column is normalized in the Euclidean norm. The resulting
reference remains available at every spatial DOF throughout
the optimization sequence.

For a transient problem, each prolonged vector $\phi_j$ defines
\begin{equation}\label{eq:mesh_reference_compression}
 H_j=(I\otimes\phi_j)^\top M(I\otimes\phi_j),\qquad
 D_j=(I\otimes\phi_j)^\top\diag(M)(I\otimes\phi_j),\qquad
 H_jq=\nu D_jq.
\end{equation}
These small temporal matrices retain spatial transport, capacity and
objective weights. The generalized eigenproblem selects temporal
directions using the compressed Jacobi diagonal.
The construction selects the requested total rank from these spatial--temporal
candidates. The prolonged candidates are approximate
fine-mesh modes; the restriction analysis retains their reference residual.

The mesh implementation keeps the reference factors on the graphics processing
unit (GPU) and evaluates
the selected products at each inactive set. Its recycling comparator keeps
selected history vectors on the same device. Both methods use the
orthogonalization, coarse solves, and residual checks described below.

\subsubsection{Analytical Cartesian spatial reference}
For the Cartesian steady problem, $M_0=I+\alpha L_s^2$, where $L_s$ is the
constant-coefficient Dirichlet Laplacian on the computational grid.
For indices $j,k\in\{1,\ldots,n_s\}$, its one-dimensional factors are
\begin{equation}\label{eq:analytical_modes}
 v_k(j)=\sqrt{\frac{2}{n_s+1}}\sin\frac{\pi jk}{n_s+1},
 \qquad
 \mu_k=\frac4{h^2}\sin^2\frac{\pi k}{2(n_s+1)}.
\end{equation}
A spatial product
$\phi_{\boldsymbol k}=v_{k_1}\otimes v_{k_2}\otimes v_{k_3}$
has Laplacian eigenvalue $\mu_{\boldsymbol k}=\sum_{a=1}^3\mu_{k_a}$
and reference eigenvalue $1+\alpha\mu_{\boldsymbol k}^2$.
We retain the $r$ lowest products. Lexicographic ordering fixes the
choice within repeated eigenvalues. The reference is homogeneous;
the CHT matrix retains the conductivity interface and prescribed
transport. Evaluating the selected products costs $O(N_sr)$ work
and storage.

\subsubsection{Mode-dependent temporal reference}
The timed Cartesian trajectories keep all temporal factors of $S$ spatial modes
(Table~\ref{tab:procedures}), which pairs each retained mode with every time level. The
ordered selection below applies when one total rank truncates the temporal directions,
as in the transformer, Bore~1 and Bore~2 trajectories and the Cartesian comparisons with
one total rank.
The Cartesian transient reference uses the same spatial Laplacian and a constant
reference capacity $c_0=1$:
\begin{equation}\label{eq:cht_space_time_reference}
 \mathcal A_0=I_{N_t}\otimes L_s+c_0L\otimes I_{N_s},
 \qquad M_0=W+\alpha\mathcal A_0^\top W\mathcal A_0.
\end{equation}
For a spatial mode $\phi_j$ with eigenvalue $\mu_j$, define
\begin{equation}\label{eq:cht_temporal_eigenproblem}
 H_j=W_t+\alpha(\mu_j I+c_0L)^\top W_t(\mu_j I+c_0L),
 \qquad H_jq_{j,k}=\nu_{j,k}q_{j,k}.
\end{equation}
Then $q_{j,k}\otimes\phi_j$ is an eigenvector of
\eqref{eq:cht_space_time_reference} with eigenvalue $\nu_{j,k}$.
The temporal vector depends on the spatial eigenvalue. We order the
candidate pairs by $\nu_{j,k}$ and retain a total
space--time rank $r$.

The construction considers the first $\min(r,N_s)$ spatial modes.
This pool contains the lowest $r$ reference eigenpairs:
$W_tL+L^\top W_t$ is positive definite for backward Euler with the
stated quadrature, so $H_j$ increases in the positive-definite order
as $\mu_j$ increases. In particular, its smallest eigenvalue increases.
Each preceding spatial mode contributes at least one eigenvalue
no larger than the first eigenvalue of a later spatial mode.
A fixed ordering resolves repeated eigenvalues.

If the selected pairs use $r_s$ distinct spatial modes, the reference
stores $N_sr_s+N_tr$ floating-point values and $r$ spatial-column
indices. Restriction evaluates a selected product at its inactive
space--time entries. The resulting dense candidate still has
$|I|r$ values. Compact reference storage therefore reduces the
representation before restriction; the online coarse solve uses
the complete restricted basis.

The transformer, Bore~1 and Bore~2 comparisons apply the same ordered selection to
the candidates of Section~\ref{sec:mesh_reference} at a fixed total rank; the Bore~3
trajectory keeps $S=50$ spatial modes, each with every temporal factor
(Table~\ref{tab:procedures}). The Cartesian
trajectories of Section~\ref{sec:v20_results} keep $S$ spatial modes, each with
one temporal factor per slab (Section~\ref{sec:final_procedure}).

\subsection{Coarse correction and projected-direction CG}
\label{sec:method}
\begin{algorithm}[!htbp]
\caption{Projected-direction CG solves the inactive equations after a coarse correction.}
\label{alg:revision_deflation}
\begin{algorithmic}[1]
\Require SPD $B$, right-hand side $b$, retained basis $Z$, initial guess $x_0$, internal tolerance $\varepsilon_c$, final tolerance $\varepsilon$ and iteration cap
\State Set $\tau=\varepsilon_c\|b\|_2$ for $b\ne0$, and $\tau=\varepsilon_c$ for $b=0$
\State Form $BZ$ and factor $E=Z^\top BZ$; define the actions of $Q$ and $P$ in \eqref{eq:revision_projectors}; set $\mathcal J=\diag(B)^{-1}$
\State $x\gets x_0+Q(b-Bx_0)$; $r\gets b-Bx$
\If{$\|r\|_2\leq\tau$} \State \Return $\operatorname{Verify}(x,\varepsilon)$ \EndIf
\State $z\gets P\mathcal Jr$; $p\gets z$; $\rho\gets r^\top z$
\For{each iteration up to the iteration cap}
  \State $w\gets Bp$; check positive curvature $p^\top w$ and $\rho$
  \State $a\gets\rho/(p^\top w)$; $x\gets x+ap$; $r\gets r-aw$
  \If{candidate convergence or scheduled residual replacement}
    \State Recompute $r\gets b-Bx$
    \If{$\|r\|_2\leq\tau$} \State \Return $\operatorname{Verify}(x,\varepsilon)$ \EndIf
    \State $x\gets x+Qr$; recompute $r\gets b-Bx$
    \If{$\|r\|_2\leq\tau$} \State \Return $\operatorname{Verify}(x,\varepsilon)$ \EndIf
    \State $z\gets P\mathcal Jr$; $p\gets z$; $\rho\gets r^\top z$
  \Else
    \State $z\gets P\mathcal Jr$; $\rho_{\rm new}\gets r^\top z$
    \State $p\gets z+(\rho_{\rm new}/\rho)p$; $\rho\gets\rho_{\rm new}$
  \EndIf
\EndFor
\State \Return $\operatorname{Verify}(x,\varepsilon)$ with the recorded termination status
\end{algorithmic}
\end{algorithm}

Rank-revealing orthogonalization of $\widetilde Z$ gives a basis
$Z\in\mathbb R^{|I|\times r_I}$ with orthonormal columns and
$r_I\leq\min(|I|,r)$. The GPU implementation performs a thin QR
factorization into orthonormal and upper-triangular factors, followed by
a singular value decomposition (SVD) of the small triangular factor, retaining
singular values greater than $10^{-12}$ times the largest. It then forms
\begin{equation}\label{eq:revision_projectors}
 E=Z^\top BZ,\qquad Q=ZE^{-1}Z^\top,\qquad P=I-QB.
\end{equation}
The matrices $Q$ and $P$ are applied as operators and never formed; a
Cholesky factorization of $E$ applies the coarse inverse. The solver uses
rank zero if the coarse condition number exceeds $10^{10}$ and records
this fallback. Rank-zero execution reduces to Jacobi-CG.

Let $D=\diag(B)$, $\mathcal J=D^{-1}$ and $x_0$ be the initial inactive state.
The coarse-corrected initial guess is
$x_c=x_0+Q(b-Bx_0)$. The projected recurrence applies $P\mathcal J$ to each
residual and solves the original system $Bx=b$
\citep{tang2009comparison,gaul2014framework}.
Algorithm~\ref{alg:revision_deflation} states the main operations.
The coarse space stays fixed during a linear solve and its residual
restarts. RefDef and recycling implementations both retain $BZ$,
so the projection evaluates $Z^\top B(\mathcal Jr)=(BZ)^\top \mathcal Jr$ through
the cached product. Its construction and storage enter the measurements.

The algebraic identities $QBQ=Q$, $P^2=P$ and $BP=P^\top B$
justify this recurrence. The corrected residual satisfies
$Z^\top(b-Bx_c)=0$, and $Z^\top BP=0$ makes every projected
search direction $B$-orthogonal to the coarse space.
For any SPD inverse preconditioner $\mathcal J$, the balanced inverse
\begin{equation}\label{eq:balanced_inverse}
 C_{\rm bal}=P\mathcal JP^\top+Q
\end{equation}
is SPD \citep{mandel1996balancing,tang2009comparison}.
On a residual with $Z^\top r=0$, it satisfies
$C_{\rm bal}r=P\mathcal Jr$. The balanced-preconditioner equivalence therefore
provides the symmetric positive-definite interpretation of the
projected recurrence with its corrected initial guess.
At each finite-precision restart, the implementation recomputes the
residual and applies the coarse correction again.

$\operatorname{Verify}(x,\varepsilon)$ evaluates $\|b-Bx\|_2/\|b\|_2$ using the
original matrix and the returned state. It uses $\|b-Bx\|_2$ when
$b=0$. The outer method separately checks the full KKT conditions. 
\Needspace{14\baselineskip}
\subsection{Restriction analysis}
\label{sec:theory}

\subsubsection{Restriction residual}
The usefulness of the supplied entries depends on their relation
to the current operator. Let $M=M_0+\Delta$, let $\Phi$ have
orthonormal columns, and define
$G_0=M_0\Phi-\Phi\Lambda$ for a symmetric small matrix $\Lambda$.
Exact reference eigenvectors have diagonal $\Lambda$ and $G_0=0$.
For approximate reference columns one may take
$\Lambda=\Phi^\top M_0\Phi$.
For a prolonged or product reference, $\Phi$ denotes an orthonormal
basis of its full-domain span in this analysis.
\begin{samepage}
With $\Phi_I=R_I\Phi$ and $\Phi_J=R_J\Phi$, block multiplication gives
\begin{equation}\label{eq:restriction_identity}
 B\Phi_I-\Phi_I\Lambda
 =-(M_0)_{IJ}\Phi_J+\Delta_{II}\Phi_I+(G_0)_I.
\end{equation}
\end{samepage}%
The three terms describe active--inactive coupling, the difference
between the current and reference operators, and reference
approximation. In the principal CHT comparisons, $M$ and $M_0$
remain fixed through the query sequence; only their restrictions
change. Thus $\Delta$ describes a fixed reference mismatch in those
experiments.

Suppose $\Phi_I=ZR$ has full column rank, with a thin QR
factorization. For $E=Z^\top BZ$ and $F=BZ-ZE$,
\begin{equation}\label{eq:restriction_residual}
 F=(I-ZZ^\top)
 \bigl[-(M_0)_{IJ}\Phi_J+\Delta_{II}\Phi_I+(G_0)_I\bigr]R^{-1}.
\end{equation}
The inverse $R^{-1}$ exposes the effect of near-dependence after
restriction. This identity motivates rank-revealing orthogonalization
and separates the sources of the subspace residual.

\subsubsection{Conditional bounds and measured convergence}
For $1\leq r<|I|$, let $U_r$ span the lowest $r$ eigendirections of
$B$ and $U_\perp$ the remaining ones, with $BU_\perp=U_\perp\Lambda_\perp$, and suppose the
retained $Z$ has $r$ columns. Premultiplying $F=BZ-ZE$ by $U_\perp^\top$ gives the Sylvester
equation
\[
 \Lambda_\perp X-XE=U_\perp^\top F,\qquad X=U_\perp^\top Z,
\]
whose coefficients $\Lambda_\perp$ and $E$ have separated spectra when
$\eta=\lambda_{r+1}(B)-\lambda_{\max}(E)>0$. Its solution then satisfies
\begin{equation}\label{eq:restriction_angle}
 \|\sin\Theta(U_r,Z)\|_F=\|X\|_F\leq\frac{\|F\|_F}{\eta}.
\end{equation}
Here the cosines of the principal angles are the singular values
of $U_r^\top Z$. The sufficient separation identifies the low
cluster. The positivity of $\eta$ and the rank $r$ are hypotheses.
Combining \eqref{eq:restriction_residual} with this estimate bounds
the largest angle through the coupling, operator and reference
residual terms. The deflated-CG perturbation estimate of
\citet{kahl2017deflated} then gives a conditional effective
condition-number bound.

For Jacobi-preconditioned CG, the relevant matrix is
$\widehat B=D^{-1/2}BD^{-1/2}$ and the relevant coarse space is
$\operatorname{range}(D^{1/2}Z)$. All preconditioned spectral
diagnostics use these variables. We measure the nonzero spectrum of
\begin{equation}\label{eq:measured_deflated_operator}
 \widehat B-
 D^{-1/2}BZ(Z^\top BZ)^{-1}Z^\top BD^{-1/2}
\end{equation}
on the small replay systems.
For an independent solution $x_*$ and zero initial guess, the
coarse correction $Qb$ removes the energy fraction
\begin{equation}\label{eq:coarse_energy_fraction}
 \chi=\frac{\|Qb\|_B^2}{\|x_*\|_B^2},
 \qquad \|v\|_B^2=v^\top Bv.
\end{equation}
Its complementary fraction is
$\|x_*-Qb\|_B^2/\|x_*\|_B^2=1-\chi$.
These measurements connect the supplied reference values to error
removal, iteration work and measured cost on optimization traces. 
\subsection{Experimental protocol}
\label{sec:benchmarks}

\subsubsection{Optimization sequences}
The body-fitted sequences contain 16 distinct targets. Two moving
Gaussian profiles define each desired state, and a pulse changes its
amplitude during a transient trajectory.
The six primary cases comprise steady and four-slab transient
transformer problems, together with the corresponding bore-in-block problems at
two nested refinement levels. They use rank 100 and five complete
repetitions per solver. Section~\ref{sec:v20_results} adds larger meshes,
longer trajectories and changing operating conditions under the procedure of
Section~\ref{sec:final_procedure}.

The operating sequences use one reference built at nominal conditions. The
transformer day chains 24 hourly eight-slab windows through a residential load
cycle, with the winding source proportional to the squared per-unit load. The
overload sequence applies four emergency load profiles to the eight-slab transformer: a
step to 1.3 per unit, a ramp to 1.5 per unit, a 1.8 per-unit pulse and an alternation
between 0.8 and 1.3 per unit. Each profile is paired with one target and four bounds
from 0.7 to 0.85, which gives 16 queries. The drive
cycle chains 16 four-slab windows on Bore~2, with the interface heat flux and swirl
speed changing between phases. Each chained window starts from the final state of
the preceding one. The off-design sequences contain 16 queries, four targets at each
of the velocity scales 0.5, 0.75, 1.5 and 2, and reassemble the operator for every
query. The shuffled sequences permute the query order.

Each Cartesian sequence contains 16 desired states indexed by
$\theta_m=m/15$, $m=0,\ldots,15$. Four Gaussian profiles have
unequal weights and query-dependent centers. Rotation and vertical
translation change their locations across queries. Transient targets
also rotate and pulse over physical time. Hashes verify that the targets are
distinct and identical between methods and repetitions.

The Cartesian steady bound is approximately $\psi=1.950\times10^{-4}$,
and its transient bound is approximately $\psi=2.685\times10^{-4}$.
The configurations retain the full-precision values.
A direct-PDAS calibration selected each bound once; the
resulting value stays fixed under spatial refinement and changes in
the number of time slabs. Uniform steps partition the fixed horizon
$T_f=0.1$ in the performance study. Unequal steps and nonzero
initial data appear in the independent discretization checks.

\subsubsection{Four complete-optimization solvers}
The first method is ordinary Jacobi-preconditioned CG with an
efficient rank-zero path. The second is RefDef with
\eqref{eq:reference_restriction}. Both use the same operator
application, precision and residual verification. The third uses
the same deflated-CG kernel with history-based recycling. It starts
with empty history. After an inner solve meets the residual criterion, its candidate
pool contains the deployed coarse vectors and the last $w$ new
CG search directions, with $w=r$. The directions are copied during
the solve, and the space remains fixed until that solve ends.

For physical candidate matrix $V$ and current Jacobi diagonal $D$,
recycling forms an orthonormal basis $\widehat V$ of $D^{1/2}V$.
It selects the lowest Ritz vectors of
$\widehat V^\top D^{-1/2}BD^{-1/2}\widehat V$ and maps them back
by $D^{-1/2}$. QR and a small-factor SVD use a relative rank threshold
of $10^{-12}$. At most $r$ physical vectors are retained.
Zero extension transfers them to the next inactive set.
Capture, selection, transfer and cached products all enter the
cost and memory measurements. A failed target clears the history.

The fourth method is AmgX-preconditioned CG with classical AMG. Configuration and
resource handles persist throughout a sequence, and each inactive matrix receives
new matrix and vector handles, a solver and an AMG hierarchy.

\subsubsection{Warm starts and accuracy}
The primary policy warm-starts the outer and inner solves.
The first query starts with all state DOFs active.
Subsequent queries receive the preceding solution and active
set after the KKT check. Within PDAS, each verified inactive solution updates
the full state used by the next inner solve. Every method has equivalent
access to this solution history. The cold-start variant resets the
active set to all DOFs for each query and uses zero
inner initial guesses. Each complete repetition starts without
solution or recycling history.

All computations use double precision. The final inactive-system
criterion is an independently recomputed relative residual of
$10^{-10}$. CG triggers candidate verification at a relative
recurrence residual of $10^{-11}$; AmgX uses a native absolute
target of $10^{-11}\|b\|_2$. For a zero right-hand side, the
corresponding absolute targets apply. The common independent
check determines acceptance. Each Cartesian inner solve has a cap of
10000 iterations. The body-fitted study uses a cap of 50000 iterations.
PDAS has a cap of 100 steps in both studies, raised to 400 for Bore~4.

Set $g=MY-\ell$ and $s=\max(1,\|\ell\|_\infty)$.
PDAS requires each of the following quantities to be at most
$10^{-8}$:
\begin{equation}\label{eq:cht_kkt_acceptance}
\begin{gathered}
 \|g+\lambda\|_\infty/s,\qquad
 \|[Y-\psi]_+\|_\infty,\qquad
 \|[-\lambda]_+\|_\infty/s,\\
 \|\lambda\odot(\psi-Y)\|_\infty/s,\qquad
 \|Y-\min(\psi,Y-g)\|_\infty/s.
\end{gathered}
\end{equation}
The primal-feasibility quantity is absolute. The other four use
the stated common normalization. Transient runs additionally
check \eqref{eq:control_adjoint} through independent slabwise
forward and transpose actions. Caps, cycles, memory limits,
exceptions and residual failures are recorded separately. The reported
times contain only sequences that satisfy these criteria at every target.

\subsubsection{Time, memory and reproducibility}
Each method and repetition runs in a fresh process on one GPU
with four central processing unit (CPU) threads. A common small warmup initializes both
GPU libraries. The method order rotates between repetitions.
The complete sequence interval starts before operator assembly
and reference or resource construction. It includes every target,
all PDAS updates and inactive solves, matrix and basis transfers,
verification and cleanup. Disjoint timing components sum to this
measured total. Reference construction is charged once, to the first
query of each sequence.

Calibration, process initialization, warmup and termination are recorded
separately; the body-fitted models use fixed bounds and need no calibration. A
separate process samples host and GPU memory.

Each case compares RefDef with the fastest tested GPU iterative
alternative: among Jacobi-CG, recycling and AmgX, the method with the lowest
median complete time over its converged repetitions, provided at least one
repetition converges at every target. We call it the fastest alternative, and the
speedup divides its median time by that of RefDef. The main-text
comparisons are those of Tables~\ref{tab:mesh_complete} and
\ref{tab:v20_scale}--\ref{tab:v20_operation}. The direct and block-in-time CPU
solvers are compared separately; the block-in-time solver applies CG with block-Jacobi
sweeps over the slab blocks of the trajectory Hessian. Most comparisons run all methods
and repetitions within one job on one GPU; the table captions identify comparisons
assembled from separate jobs. Repeated measurements of one Bore~3 configuration in
separate jobs show between-job timing variation. The reported envelopes describe the
sampled repetitions, so small speedup margins should be interpreted within that
measurement scope.

Source identifiers, configurations, target and matrix hashes, compressed
active sets and scripts that regenerate every table and figure accompany the
timing records.

\subsubsection{Timing procedures}
\label{sec:final_procedure}
Two procedures share Algorithm~\ref{alg:revision_deflation} and the accuracy
criteria above (Table~\ref{tab:procedures}). The
per-update assembly procedure forms each inactive Hessian by sparse products after
every active-set update; the primary comparisons of Section~\ref{sec:mesh_results}
use it, as the original body-fitted implementation does. The sliced-Hessian procedure
of Algorithm~\ref{alg:final_procedure} assembles the full reduced Hessian once for
each operator and extracts every inactive principal submatrix from it; the comparisons
of Section~\ref{sec:v20_results} use it. The extraction removes the per-update sparse
products from the setup of every method.
Reference meshes, ranks and temporal coverage are prescribed settings of each
comparison. One reference, computed before the first query at the nominal flow, is
used for every query of a sequence.

Targets, loads, bounds and chained initial states change only the right-hand side
or the constraints and leave the full reduced Hessian unchanged; each active-set update
extracts its current inactive principal submatrix. A change of flow scale reassembles
the full reduced Hessian. The comparisons use Algorithm~\ref{alg:revision_deflation}
with the coarse correction and residual replacement. The \href{https://github.com/teeratornk/deflation-example/blob/ba682dd05b168d2f43ff0c58050813deee785223/docs/mesh-showcases.md#choose-the-numerical-procedure}
{reproduction examples} give the commands for each procedure.

\begin{table}[pos=htbp]
\centering\footnotesize\setlength{\tabcolsep}{3pt}
\caption{Timing procedures and their prescribed reference settings. Both procedures
solve each inactive system with Algorithm~\ref{alg:revision_deflation}, including its
residual replacement, and restrict the reference on the GPU. Complete temporal coverage
keeps every temporal factor of each retained spatial mode; truncated coverage keeps the
lowest spatial--temporal pairs up to the total rank. Temporal coverage does not apply to
steady cases.}
\label{tab:procedures}
\begin{tabularx}{\textwidth}{@{}>{\raggedright\arraybackslash}p{0.15\textwidth}>{\raggedright\arraybackslash}p{0.12\textwidth}>{\raggedright\arraybackslash}p{0.15\textwidth}>{\raggedright\arraybackslash}X>{\raggedright\arraybackslash}p{0.22\textwidth}@{}}
\toprule
Procedure & Comparisons & Inactive Hessian & Reference mesh and rank & Temporal coverage \\
\midrule
Per-update assembly & Table~\ref{tab:mesh_complete} & Sparse products at every update &
Packaged mesh, total rank 100 & Truncated with four slabs; the steady cases have no
temporal factors \\
Sliced Hessian & Table~\ref{tab:v20_scale} (steady) & Principal submatrix of the
preassembled Hessian & Closed form on Cartesian grids, rank 100; packaged mesh for the
transformer and Bore~1, rank 100; first refinement for Bore~2 and Bore~3 and second
refinement for Bore~4, rank 200 & No temporal factors \\
Sliced Hessian & Table~\ref{tab:v20_transient} (trajectories) & Principal submatrix of
the preassembled Hessian & Closed form on Cartesian grids; packaged mesh for the
transformer, Bore~1 and Bore~2, total rank 100; first refinement for Bore~3, rank 200 &
Complete on Cartesian grids ($S=12$ at eight slabs, $S=6$ at 16 and 32) and Bore~3
($S=50$); truncated elsewhere \\
Sliced Hessian & Table~\ref{tab:v20_operation} (steady and trajectory sequences) &
Principal submatrix of the preassembled Hessian & Packaged mesh, total rank 100; first
refinement for Bore~3, rank 200 & Truncated for the transformer day, overloads, drive
cycle and four-slab shuffled sequences. The steady sequences have no temporal
factors \\
\bottomrule
\end{tabularx}
\end{table}

\begin{algorithm}[htbp]
\caption{The sliced-Hessian procedure builds one reference and reuses it for every inactive-set solve of a query sequence.}
\label{alg:final_procedure}
\begin{algorithmic}[1]
\Require Queries $q=1,\dots,N_q$ with target, bound, load and flow scale; prescribed reference mesh and rank $r$, or $S$ spatial modes with all $N_t$ temporal factors each
\State Assemble the reduced Hessian at the nominal flow
\State Form the closed-form Laplacian modes on a Cartesian grid, or compute the lowest generalized eigenvectors of the coarse reduced Hessian at the nominal flow with residuals at most $10^{-7}$ and prolong them to the target mesh
\State For a trajectory, keep the $r$ lowest spatial--temporal pairs, or all $N_t$ temporal factors of each of $S$ spatial modes
\State Store the reference on the GPU
\For{$q=1,\dots,N_q$}
  \If{the flow scale of query $q$ differs from that of the current operator}
    \State Reassemble the transport and the reduced Hessian; keep the reference
  \EndIf
  \State Form the right-hand side from the target and load of query $q$; for a chained window, take the initial state from the final state of window $q-1$
  \State Start from the previous converged state and active set; the first query starts with every constraint active
  \Repeat
    \State Extract the inactive principal submatrix of the current Hessian and restrict the reference on the GPU
    \State Solve with Algorithm~\ref{alg:revision_deflation}; verify the relative residual at most $10^{-10}$
    \State Update the multipliers and the active set for the bound of query $q$
  \Until{every KKT measure is at most $10^{-8}$}
\EndFor
\end{algorithmic}
\end{algorithm} 

\sisummary The supporting information (SI) extends this section. SI
Section~\Sref{si:methodology} gives the detailed configurations:
Section~\Sref{si:mesh_formulation} the body-fitted element formulation,
Section~\Sref{si:transport_forms} the transport forms with the energy-stability argument
and its verification, Section~\Sref{si:mesh_targets} the targets and the physical scales
of both models, Section~\Sref{si:transient} the blockwise temporal treatment of the
space--time operator, Section~\Sref{si:mesh_reference} the temporal compression of the
body-fitted references and the matched tensor alternative, and
Section~\Sref{si:theory} the conditional effective condition-number bound with the
derivation of the angle estimate. Table~\Sref{si:tab_experiment_map} summarizes the
principal comparisons, Section~\Sref{si:cht_timing} defines every timing interval, the
speedup envelope and the memory measurements, and Table~\Sref{si:tab_amgx_configuration}
gives the AmgX configuration. Section~\Sref{si:mesh_refinement_diagnosis} describes an
additional residual correction for a finer transformer mesh, and a repeated Bore~3
configuration in Table~\Sref{si:tab_v20_comparisons} measures the between-job timing
variation. 
\Needspace{13\baselineskip}
\section{Results}
\label{sec:results}

\subsection{Independent verification}
\label{sec:cht_accuracy}
PDAS agrees with independent bounded-variable least-squares solutions of small
steady and transient problems. The body-fitted problems include the complete
coarse bore-in-block mesh and a transformer inlet-region subproblem with uniform and
unequal time steps; the Cartesian problems exercise nonzero initial temperatures,
final-time coupling, quadrature weights and recovered controls. Every primary
body-fitted sequence meets the residual and KKT criteria of
Section~\ref{sec:benchmarks} at every target.

Relaxing the upper bound admits a higher peak temperature, lowers the tracking error
and shrinks the active region in both Bore~1 problems of
Figure~\ref{fig:optimized_temperature_bounds}, where each row keeps the target,
operator and regularization fixed.

\begin{figure}[pos=htbp]
\centering
\includegraphics[width=0.75\textwidth]{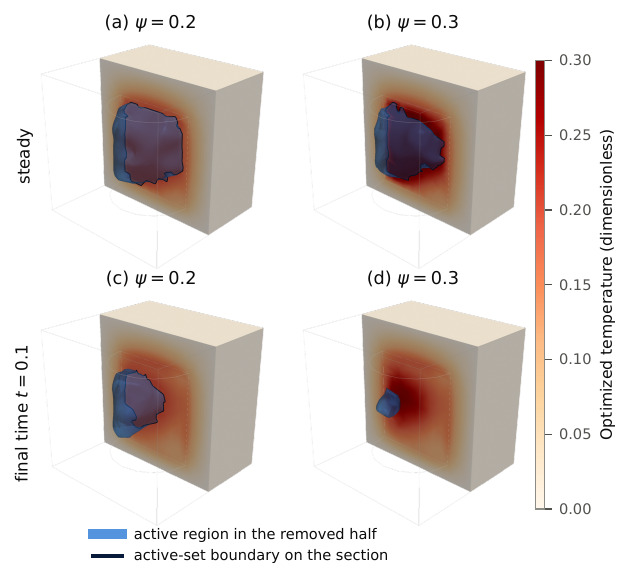}
\caption{Optimized Bore 1 temperatures under a tighter and a looser upper bound,
at a fixed desired state and operator. Each rendering removes the front half of the block
and colors the remaining half by temperature on a common scale; the translucent
blue surface is the active region in the removed half, and the dark blue line
is its boundary on the exposed section $x_2/L=0.5$. The upper row shows steady
solutions and the lower row the final state at $t=T_f=0.1$ of complete
four-slab optimizations; four slabs do not resolve the initial thermal layer of these
trajectories. The fields come from RefDef and agree
with sparse-direct PDAS solutions.}
\label{fig:optimized_temperature_bounds}
\end{figure}

\subsection{Complete body-fitted optimization cost}
\label{sec:mesh_results}
RefDef has the lowest median complete-sequence time in all six
body-fitted cases of Table~\ref{tab:mesh_complete}. The four-slab cases are timed as
discrete all-at-once problems; on Bore~1, the tested slab counts do not resolve the
initial thermal layer (Section~\ref{sec:v20_results}).

\begin{table}[pos=htbp]
\centering\footnotesize\setlength{\tabcolsep}{4pt}
\caption{Median complete-sequence times in seconds over five repetitions of 16-target
optimization sequences; every method converges at every target in every repetition.
The speedup is the ratio of the median times of the fastest alternative and RefDef; brackets give the envelope of ratios between extreme repetition times. Dimension counts the optimized spatial or space--time
DOFs. RefDef and recycling use total rank 100, and the transformer uses the
skew-symmetric transport \eqref{eq:skew_transport}.}
\label{tab:mesh_complete}
\begin{tabular}{@{}llrrrrrl@{}}
\toprule
Geometry & Form & Dimension & Jacobi-CG & RefDef & Recycling & AmgX & Speedup \\
\midrule
Transformer & Steady & 10830 & 24.961 & 10.532 & 35.666 & 17.535 & 1.66 [1.62, 1.71] \\
Transformer & 4 slabs & 43320 & 38.564 & 22.543 & 73.053 & 51.116 & 1.71 [1.70, 1.72] \\
Bore 1 & Steady & 3958 & 6.358 & 3.805 & 5.698 & 8.740 & 1.50 [1.45, 1.51] \\
Bore 1 & 4 slabs & 15832 & 11.752 & 10.113 & 13.942 & 34.289 & 1.16 [1.13, 1.19] \\
Bore 2 & Steady & 36197 & 57.602 & 32.481 & 48.158 & 102.570 & 1.48 [1.42, 1.52] \\
Bore 2 & 4 slabs & 144788 & 162.978 & 142.336 & 176.113 & 527.429 & 1.15 [1.09, 1.28] \\ 
\bottomrule
\end{tabular}
\end{table}

The iteration savings exceed the additional reference setup cost
(Figure~\ref{fig:mesh_sequences}(a)). RefDef has the shortest iteration
interval in every case; on the four-slab Bore~2 trajectory, the largest case,
Figure~\ref{fig:mesh_sequences}(b) shows it ahead from the first target, which carries
its construction.

\begin{figure}[pos=htbp]
\centering
\includegraphics[width=\textwidth]{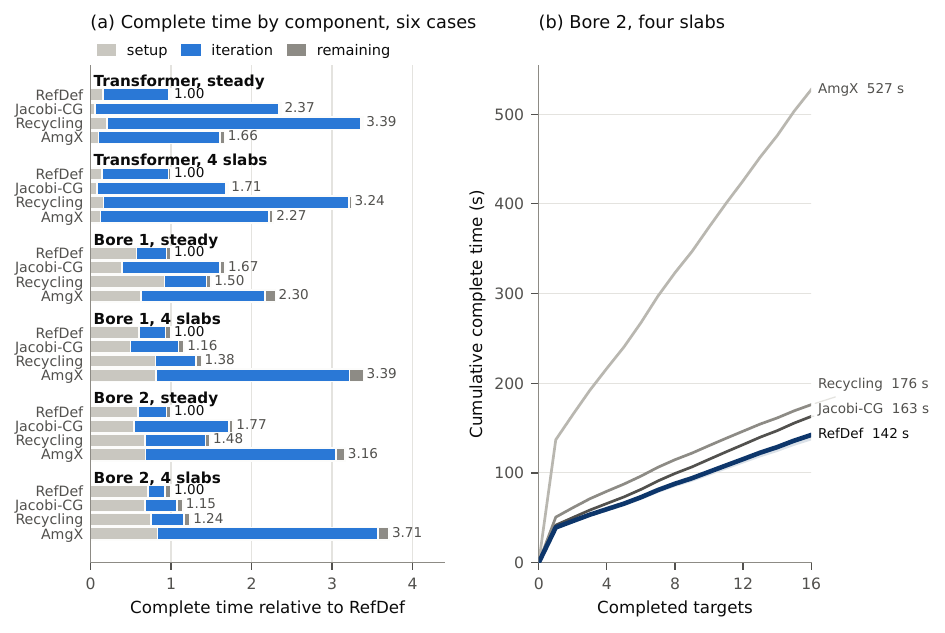}
\caption{Complete body-fitted costs. (a) Complete time of each method relative
to RefDef for the six cases of Table~\ref{tab:mesh_complete}, split
into setup (assembly, reference or resource construction, inactive extraction,
basis processing and coarse or hierarchy setup), iteration, and remaining work
(transfers, verification and cleanup) of the median-time repetition; labels give
the relative complete time. (b) Median cumulative time over the 16 targets of
the four-slab Bore~2 sequence; the band spans the five RefDef repetitions,
and the first target includes assembly and reference or resource construction.}
\label{fig:mesh_sequences}
\end{figure} 

\subsection{Reference information at released constraints}
\label{sec:cht_transfer_results}
Direct restriction keeps reference directions at the DOFs that constraints release,
and the retained directions improve coarse-error removal on the tested traces. The
matched replays of Section~\ref{sec:reference} isolate this retention from the rest of
the method: both transfers start from the same reference restriction and solve the same
systems from zero initial guesses. The complete recycling method also learns new
directions from its inner solves; the replays exclude this learning.

On every Cartesian trace, direct restriction removes more of the initial energy error
than zero-extension transfer, and the replayed solves need fewer iterations and less
time (Table~\ref{tab:cht_transfer}). Both transfers retain the same
rank on every trace except the steady $24^3$ one, so the difference comes from which
directions are retained. Figure~\ref{fig:cht_transfer} locates it: DOFs released after
the first restriction carry no reference entries under zero-extension transfer, the
released region grows as the targets move, and the gap in removed energy persists over
later queries. On the Bore~2 traces, direct restriction also removes more initial energy
error and needs fewer iterations; the transformer trace, on which few constraints change
between solves, shows little difference.

\begin{table}[pos=htbp]
\centering\small
\caption{Matched replays of direct restriction and zero-extension transfer of the
full-domain reference on the same systems with zero initial guesses. Time
and iterations sum per-kernel medians over three repetitions. Solves counts the
replayed inactive systems, Rank gives the smallest and largest retained rank over these
solves, and Energy fraction is the median of
Equation~\eqref{eq:coarse_energy_fraction}.}
\label{tab:cht_transfer}
\begin{tabular}{@{}lrllrrr@{}}
\toprule
Grid and slabs & Solves & Transfer & Rank & Energy fraction &
Iterations & Time (s) \\
\midrule
$ 6^3$ & 32 & Direct restriction & 20--20 & 0.975 & 1002 & 0.048 \\
$ 6^3$ & 32 & Zero-extension transfer & 20--20 & 0.058 & 1625 & 0.072 \\
$ 4^3$$\times 4$ & 35 & Direct restriction & 20--20 & 0.953 & 1354 & 0.066 \\
$ 4^3$$\times 4$ & 35 & Zero-extension transfer & 20--20 & 0.057 & 1930 & 0.087 \\
$ 24^3$ & 86 & Direct restriction & 100--100 & 0.651 & 19400 & 6.034 \\
$ 24^3$ & 86 & Zero-extension transfer & 93--100 & 0.035 & 81533 & 25.404 \\
$ 12^3$$\times 8$ & 65 & Direct restriction & 100--100 & 0.403 & 17521 & 5.748 \\
$ 12^3$$\times 8$ & 65 & Zero-extension transfer & 100--100 & 0.025 & 35998 & 13.268 \\
$ 16^3$$\times 4$ & 74 & Direct restriction & 100--100 & 0.314 & 16749 & 6.166 \\
$ 16^3$$\times 4$ & 74 & Zero-extension transfer & 100--100 & 0.016 & 57025 & 19.127 \\ 
\bottomrule
\end{tabular}
\end{table}

\begin{figure}[pos=htbp]
\centering
\includegraphics[width=\textwidth]{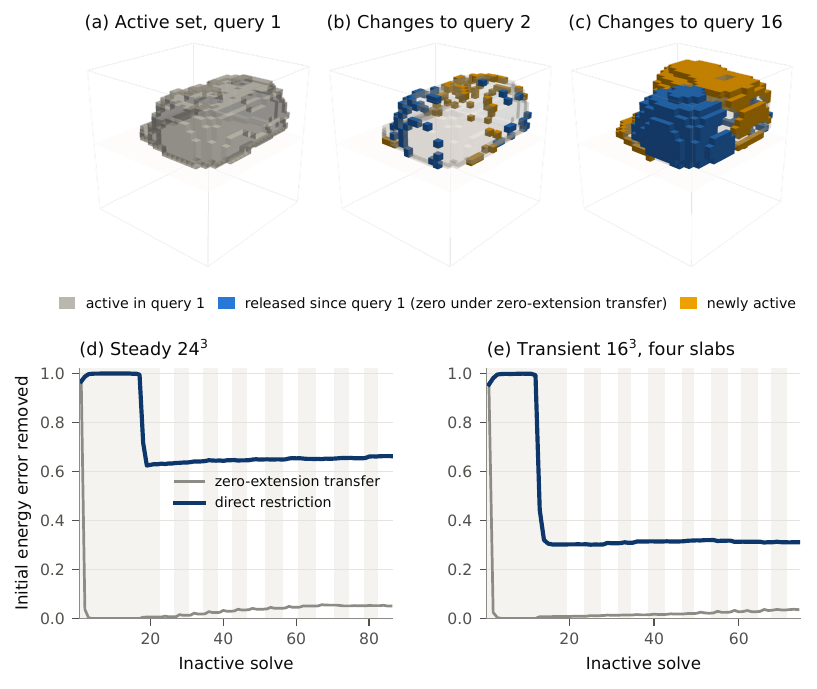}
\caption{Released constraints and the reference information retained by each
transfer. (a) Converged active set of the first query of the steady $24^3$
sequence above the solid--fluid interface (orange plane). (b, c) Changes
relative to query 1 at queries 2 and 16: released DOFs (blue) have zero entries
under zero-extension transfer and full-reference entries under direct
restriction; newly active DOFs are amber. (d, e) Initial energy error removed by
each coarse correction in the matched replays of the steady $24^3$ sequence and
the transient $16^3$ sequence with four slabs; shaded bands mark alternate
queries.}
\label{fig:cht_transfer}
\end{figure}

On the two smallest traces, the complete spectrum shows the same ordering: direct
restriction lowers the effective condition number of the deflated operator far more
than zero-extension transfer.

The conditional analysis of Section~\ref{sec:theory} describes this behavior when its
rank and separation assumptions hold. On a controlled heterogeneous $6^3$ example with
an exact lowest eigenvector of the full Jacobi-scaled matrix, the sufficient bound
certifies a lower effective condition number for small restrictions and becomes
uninformative as more nodes are removed.

The retained information also lowers the complete optimization cost. A complete-sequence
ablation times the primary cases of Table~\ref{tab:mesh_complete} with the same
reference carried between inactive sets by zero extension, in the same job as direct
restriction and Jacobi-CG. Direct restriction keeps its advantage over Jacobi-CG on the
bore-in-block cases, and that advantage disappears with zero extension. On the transformer, zero extension still
outperforms Jacobi-CG; the four-slab trajectory shows the smallest difference between
the two transfers, with overlapping timing ranges. Equation~\eqref{eq:full_reference_transfer}
thus identifies the entries lost in transfer, the replays measure their effect on the
coarse correction, and the ablation shows their effect on the complete cost.

\subsection{Scale, trajectories and operating conditions}
\label{sec:v20_results}
Coarse-mesh reference information carries over to refined meshes. On the steady
problems of Table~\ref{tab:v20_scale}, the iteration savings over Jacobi-CG grow with
Cartesian refinement, and the refined bore-in-block meshes keep their advantage up to
Bore~4 with a reference built on a coarser level (Table~\ref{tab:procedures}). The
Bore~4 comparison includes Jacobi-CG and AmgX from separate jobs, and a Bore~4
reference built on the third refinement makes RefDef slower than AmgX. On the transformer refined twice, RefDef fails at a
steady query that AmgX solves.

\begin{table}[pos=htbp]
\centering\footnotesize\setlength{\tabcolsep}{4pt}
\caption{Steady comparisons under Algorithm~\ref{alg:final_procedure} with 16 targets.
Times are median complete-sequence seconds over the repetitions, out of five, that
converge at every target. The speedup is the ratio of the median times of the fastest
alternative and RefDef, and brackets give the envelope of ratios between
extreme repetition times; the iteration ratio divides
the median inner iterations of Jacobi-CG by those of RefDef, and a dash
marks a ratio without a converged Jacobi-CG repetition. The last column lists each
method with a repetition that fails to converge, with its converged and attempted
repetitions. The standard AmgX stopping factor is $10^{-1}$; AmgX ($10^{-2}$) and AmgX
($10^{-3}$) set it to $10^{-2}$ and $10^{-3}$. The $96^3$ and Bore~4 envelopes combine
measurements from separate jobs. $^\dagger$Compared with a subset of the alternatives: Bore~4
with Jacobi-CG and AmgX, where PDAS may take up to 400 steps.}
\label{tab:v20_scale}
\begin{tabular}{@{}llrrlrr>{\raggedright\arraybackslash}p{2.6cm}@{}}
\toprule
Problem & Size & Unknowns & RefDef & Fastest alternative & Speedup [range] & \begin{tabular}[b]{@{}r@{}}Iteration\\ratio\end{tabular} & Not converged \\
\midrule
Cartesian & $32^3$ & 32{,}768 & 13.5 & AmgX 25.2 & 1.87 [1.85, 1.91] & 4.34 & -- \\
Cartesian & $48^3$ & 110{,}592 & 40.7 & AmgX 88.5 & 2.17 [2.12, 2.37] & 4.39 & -- \\
Cartesian & $64^3$ & 262{,}144 & 100.4 & AmgX 177.5 & 1.77 [1.70, 1.83] & 4.42 & -- \\
Cartesian & $96^3$ & 884{,}736 & 514.1 & AmgX ($10^{-2}$) 715.1 & 1.39 [1.37, 1.40] & 4.56 & AmgX 0/5 \\
Cartesian & $128^3$ & 2{,}097{,}152 & 1885.0 & Jacobi-CG 3644.8 & 1.93 [1.65, 2.67] & 6.32 & AmgX ($10^{-2}$) 0/5, AmgX ($10^{-3}$) 0/5, AmgX 0/5 \\
Transformer & packaged mesh & 10{,}830 & 10.6 & AmgX 17.4 & 1.65 [1.56, 1.70] & 4.19 & -- \\
Bore-in-block & Bore 1 & 3{,}958 & 4.2 & recycling 6.0 & 1.42 [1.26, 1.59] & 4.95 & -- \\
Bore-in-block & Bore 2 & 36{,}197 & 30.0 & recycling 42.3 & 1.41 [1.31, 1.56] & 7.55 & -- \\
Bore-in-block & Bore 3 & 308{,}171 & 330.8 & recycling 560.6 & 1.69 [1.39, 1.94] & 7.03 & -- \\
Bore-in-block$^\dagger$ & Bore 4 & 2{,}540{,}695 & 7065.2 & AmgX 9851.0 & 1.39 [1.36, 1.41] & -- & Jacobi-CG 0/5 \\ 
\bottomrule
\end{tabular}
\end{table}

Temporal coverage sets how the reference serves trajectories. With complete temporal
coverage, the Cartesian references keep similar iteration savings from $12^3$ to $64^3$
(Table~\ref{tab:v20_transient}). The complete-time margin narrows as a trajectory gains
slabs at a fixed horizon, and on Bore~1 with eight slabs the timing ranges of RefDef and
Jacobi-CG overlap. On the $16^3$ Cartesian trajectory with 64 slabs, Jacobi-CG has the
lower complete time in a single repetition with either temporal coverage. Setup forms a large share of the Jacobi-CG, RefDef and
recycling times on the four-slab bore-in-block sequences of
Table~\ref{tab:mesh_complete} (Figure~\ref{fig:mesh_sequences}(a)). The sliced-Hessian
procedure removes most of this shared setup, and the same two cases show larger
speedups in Table~\ref{tab:v20_transient}.

The trajectory comparisons time the discrete all-at-once optimization, whose bounds
hold at the optimization time levels. A replay of the recovered control on finer time
steps can exceed the bound between levels. On the 600-s transformer trajectory this
exceedance is smaller with 64 slabs than with eight, and the Bore~1 trajectory begins
with a thermal layer
that the tested slab counts do not resolve.

\begin{table}[pos=htbp]
\centering\footnotesize\setlength{\tabcolsep}{4pt}
\caption{Trajectory comparisons optimize all slabs of the horizon at once. Columns
follow Table~\ref{tab:v20_scale}, and unknowns count all slabs. The Cartesian references
keep $S=12$ spatial modes, each with a temporal factor for every slab, at eight slabs and $S=6$ at 16 and 32 slabs, and the
transformer uses the skew-symmetric transport. $^\dagger$Compared with Jacobi-CG alone.}
\label{tab:v20_transient}
\begin{tabular}{@{}llrrlrr>{\raggedright\arraybackslash}p{2.6cm}@{}}
\toprule
Problem & Slabs & Unknowns & RefDef & Fastest alternative & Speedup [range] & \begin{tabular}[b]{@{}r@{}}Iteration\\ratio\end{tabular} & Not converged \\
\midrule
Cartesian & $12^3$, 8 slabs & 13{,}824 & 3.9 & Jacobi-CG 5.5 & 1.43 [1.39, 1.45] & 3.12 & -- \\
Cartesian & $16^3$, 8 slabs & 32{,}768 & 7.1 & Jacobi-CG 11.1 & 1.56 [1.51, 1.65] & 3.16 & -- \\
Cartesian & $24^3$, 8 slabs & 110{,}592 & 22.1 & Jacobi-CG 35.7 & 1.61 [1.59, 1.65] & 3.10 & -- \\
Cartesian & $32^3$, 8 slabs & 262{,}144 & 65.9 & Jacobi-CG 95.4 & 1.45 [1.42, 1.47] & 3.08 & -- \\
Cartesian & $48^3$, 8 slabs & 884{,}736 & 324.0 & Jacobi-CG 388.8 & 1.20 [1.19, 1.22] & 2.99 & -- \\
Cartesian & $64^3$, 8 slabs & 2{,}097{,}152 & 1145.9 & Jacobi-CG 1279.7 & 1.12 [1.09, 1.12] & 3.03 & AmgX 0/5 \\
Cartesian & $16^3$, 16 slabs & 65{,}536 & 12.2 & Jacobi-CG 15.7 & 1.29 [1.22, 1.30] & 2.36 & -- \\
Cartesian & $16^3$, 32 slabs & 131{,}072 & 20.6 & Jacobi-CG 22.3 & 1.08 [1.07, 1.12] & 2.52 & -- \\
Transformer & 4 slabs & 43{,}320 & 21.6 & Jacobi-CG 37.9 & 1.75 [1.69, 1.76] & 2.91 & -- \\
Transformer & 8 slabs & 86{,}640 & 27.4 & Jacobi-CG 44.3 & 1.62 [1.59, 1.64] & 2.69 & -- \\
Bore 1 & 4 slabs & 15{,}832 & 6.9 & Jacobi-CG 8.8 & 1.27 [1.17, 1.33] & 2.75 & -- \\
Bore 1 & 8 slabs & 31{,}664 & 10.2 & Jacobi-CG 11.0 & 1.07 [0.99, 1.65] & 1.97 & -- \\
Bore 2 & 4 slabs & 144{,}788 & 72.0 & Jacobi-CG 91.2 & 1.27 [1.25, 1.30] & 2.72 & -- \\
Bore 2 & 8 slabs & 289{,}576 & 140.3 & Jacobi-CG 143.5 & 1.02 [1.01, 1.05] & 1.96 & -- \\
Bore 3$^\dagger$ & 4 slabs & 1{,}232{,}684 & 1474.6 & Jacobi-CG 1644.8 & 1.12 [1.10, 1.14] & 3.90 & -- \\ 
\bottomrule
\end{tabular}
\end{table}

A reference built at nominal conditions remains effective when loads, chained windows
and query order change (Table~\ref{tab:v20_operation}), and on the bore-in-block meshes
across off-design flow rates, which reassemble the operator at every query. No tested
solver completes the transformer off-design sequence.

\begin{table}[pos=htbp]
\centering\footnotesize\setlength{\tabcolsep}{3pt}
\caption{Operating sequences change loads, flow rates, windows or query order; one
reference built at nominal conditions is used for every query. Columns follow
Table~\ref{tab:v20_scale}, and the transformer uses the skew-symmetric transport in
every sequence.}
\label{tab:v20_operation}
\begin{tabular}{@{}llrrlrr>{\raggedright\arraybackslash}p{2.4cm}@{}}
\toprule
Scenario & Problem & Unknowns & RefDef & Fastest alternative & Speedup [range] & \begin{tabular}[b]{@{}r@{}}Iteration\\ratio\end{tabular} & Not converged \\
\midrule
Transformer day & 24 hourly windows & 86{,}640 & 53.1 & Jacobi-CG 77.1 & 1.45 [1.16, 1.74] & 2.42 & -- \\
Drive cycle & 16 windows & 144{,}788 & 85.5 & Jacobi-CG 107.5 & 1.26 [1.23, 1.28] & 2.74 & -- \\
Off-design flow & transformer & 10{,}830 & \multicolumn{5}{l}{no method converges at every query} \\
Off-design flow & Bore 2 & 36{,}197 & 43.6 & recycling 60.8 & 1.39 [1.35, 1.41] & 4.86 & -- \\
Off-design flow & Bore 3 & 308{,}171 & 529.0 & recycling 763.8 & 1.44 [1.41, 1.48] & 7.21 & -- \\
Overloads & transformer, 8 slabs & 86{,}640 & 31.9 & Jacobi-CG 46.0 & 1.44 [1.20, 1.56] & 2.45 & -- \\
Shuffled order & transformer & 10{,}830 & 15.6 & AmgX 17.5 & 1.12 [1.07, 1.16] & 2.72 & -- \\
Shuffled order & Bore 2 & 36{,}197 & 27.8 & recycling 42.7 & 1.54 [1.52, 1.59] & 4.86 & -- \\
Shuffled order & Bore 2, 4 slabs & 144{,}788 & 88.3 & Jacobi-CG 116.1 & 1.32 [1.29, 1.63] & 2.73 & -- \\ 
\bottomrule
\end{tabular}
\end{table}

Figure~\ref{fig:v20_scale}(a)--(c) gathers the speedups of the three tables. The sparse
direct factorization, compared separately on the CPU, is faster on the steady
transformer.

The construction cost is shared by every query of a sequence. Over a 256-query Bore~2
sequence (Figure~\ref{fig:v20_scale}(d)), RefDef has the lowest cumulative
time from the first query, construction included, and the time it saves accumulates
with each further query.

\begin{figure}[pos=htbp]
\centering
\includegraphics[width=\textwidth]{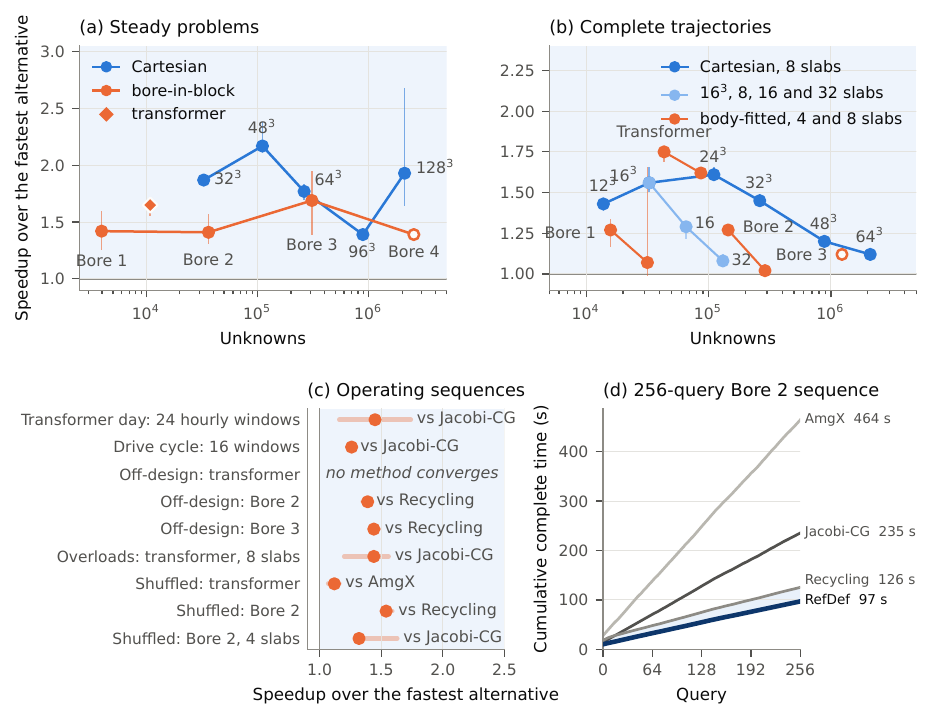}
\caption{Speedup of RefDef over the fastest alternative; shading marks
speedups above one, bars span the envelope of ratios between extreme repetition times,
and open markers denote
comparisons with a subset of the alternatives. (a) Steady problems of
Table~\ref{tab:v20_scale} against the number of unknowns, with the Cartesian grids and
bore-in-block meshes joined in order of refinement. (b) Trajectories of
Table~\ref{tab:v20_transient}: Cartesian grids with eight slabs, the $16^3$ grid with 8,
16 and 32 slabs, and the body-fitted trajectories with four and eight slabs. (c) Operating
sequences of Table~\ref{tab:v20_operation}, labeled by the fastest alternative. (d) Median
cumulative time of each method over a 256-query Bore~2 sequence; the shaded region marks
the lead of RefDef over recycling.}
\label{fig:v20_scale}
\end{figure}
\FloatBarrier

\sisummary The SI extends Sections~\ref{sec:cht_accuracy} to~\ref{sec:v20_results}. SI
Sections~\Sref{si:mesh_results} and~\Sref{si:cht_validation_results} verify the
body-fitted and Cartesian optimizations, Section~\Sref{si:v20_temporal} gives the time
refinement and replay of the trajectories, and Section~\Sref{si:temperature_bounds} the
extrema, tracking errors and active fractions of the fields of
Figure~\ref{fig:optimized_temperature_bounds}. Tables~\Sref{si:tab_mesh_transfer}
and~\Sref{si:tab_small_spectrum} give the mesh transfer replays and the spectra of the
small replays, Table~\Sref{si:tab_retention_ablation} the complete-sequence retention
ablation and Table~\Sref{si:tab_cht_conditioning} the controlled conditioning example.
Table~\Sref{si:tab_v20_comparisons} adds tolerance, regularization, cold-start, rank,
reference-level, sequence-length, load-and-bound and one-total-rank variations,
Table~\Sref{si:tab_v20_single} the single-repetition comparisons and solver limits, and
Table~\Sref{si:tab_v20_direct} the CPU factorization solvers.
Section~\Sref{si:limits} collects the settings in which the margin narrows, reverses or
is unavailable. 

\section{Discussion}
\label{sec:discussion}

\subsection{The value of full-domain reference retention}
Each PDAS update changes the inactive set and therefore the principal submatrix of the
reduced Hessian that the inner solver sees. A reference stored on the full state domain
supplies the coarse correction with entries at the DOFs that a constraint has just
released, where zero extension assigns zero. The replays of
Section~\ref{sec:cht_transfer_results} attribute the faster error removal at equal or nearly
equal rank to these entries, and the complete-sequence ablation shows that they matter for the
measured optimization cost. The larger transfer effect on the bore-in-block traces is
consistent with their more extensive constraint changes. The value of the retained
entries also depends on the reference and on the error components present in each
solve. On the steady transformer, a two-dimensional axisymmetric operator, a sparse
direct factorization is faster than RefDef.

The space--time construction carries the policy to trajectories. With complete
temporal coverage, the restricted space pairs each retained spatial mode with every
time level, and the temporal eigenvectors provide an alternative basis of the same
correction space in exact arithmetic. The
Cartesian and Bore~3 trajectories use this coverage. When the total rank truncates the
temporal directions, as on the transformer, Bore~1 and Bore~2 trajectories, the
temporal eigenvalue ordering decides which time dependence each mode keeps. On the
eight-slab Cartesian trajectories, one total rank shared by all slabs gives smaller
speedups than complete coverage. The
complete-time advantage narrows as the number of slabs increases at a fixed horizon;
this behavior depends on the coupled temporal system and on the temporal coverage of the
reference. The Cartesian trajectory comparisons retain similar iteration savings across
mesh sizes, and the complete-time advantage decreases on the larger grids. Reference
selection therefore needs to account for the cost of basis processing and coarse
correction as well as for the iteration reduction.

The conditional bounds of Section~\ref{sec:theory} hold under the stated rank and
spectral-separation assumptions, and the controlled $6^3$ example evaluates them with an
exact reference. The timed references are analytical or prolonged approximations; the
bounds explain how retained information enters the coarse correction, and the replays and
complete sequences measure its effect.

\Needspace{10\baselineskip}
\subsection{Iteration savings and complete cost}
A reference pays for itself when the iterations it saves cost more than building,
restricting and applying it. Two features of the tested procedures favor this balance.
The body-fitted reference is built once, on a mesh level at which the eigensolve is
inexpensive, and prolonged to finer target meshes, so every query shares its cost;
Figure~\ref{fig:v20_scale}(d) shows the lead at the first query and the time saved over
a 256-query sequence. Restriction takes place on the GPU, next to the operator, at
every active-set update. The prescribed reference level and rank set the balance:
finer reference levels raise the iteration savings on the bore-in-block meshes at a
higher construction cost.

Because targets, bounds and loads leave the full reduced Hessian unchanged, one
reference applies to every query of such a sequence. Flow rates change the operator itself. On the bore-in-block meshes the modes
computed at the nominal velocity remain informative across the tested off-design
flows, so RefDef with the nominal reference keeps the lowest median complete time
there.

\subsection{Scope}
\label{sec:limitations}
The application is linear-quadratic heat transfer with prescribed velocity,
temperature-independent material properties, an unrestricted distributed heat source
that may heat or cool, and discrete temperature bounds. The timings concern the tested GPU
implementations, hardware and complete-sequence boundary, and the retained basis adds
storage that grows with rank. The comparisons solve the reduced inactive systems;
full-space block preconditioners of the KKT system lie outside the tested set.
Further work will examine lower-memory reference spaces,
adaptive updates, stronger reusable preconditioners and
optimization with coupled flow and temperature-dependent properties.%

\sisummary SI Section~\Sref{si:temporal_coverage} shows that complete temporal coverage and
the temporal eigenvectors span the same correction space,
Table~\Sref{si:tab_v20_comparisons} compares one total rank with complete coverage, and
Section~\Sref{si:limits} reports the settings in which the measured advantage narrows or
reverses. 
\section{Conclusion}
\label{sec:conclusions}

The location and persistence of correction information can determine the complete
cost of repeated constrained optimization, even when the reference construction and the
inner solver are unchanged. A full-domain reference keeps a value at every state DOF as
temperature constraints activate and release, and direct restriction supplies each
inactive solve with reference entries at the DOFs that a constraint has just released.
Carrying the same reference between inactive sets by zero extension raises the median
complete time of every tested body-fitted sequence, with overlapping timing ranges on the
four-slab transformer, and on the bore-in-block models it removes the advantage over
Jacobi-CG. Full-domain reference retention is therefore suited to
repeated temperature-constrained optimization with prescribed flow, in which many queries
share one operator or a small family of operators. 
\section*{Acknowledgements}

F.~Ballarin acknowledges support from the Istituto Nazionale di Alta
Matematica (INdAM), Gruppo Nazionale per il Calcolo Scientifico (GNCS),
through the project ``Metodi numerici avanzati per problemi di fluidodinamica
computazionale'' (project code E53C25002010001).

Y.~Choi was supported by the United States Department of Energy, Office of
Science, Office of Advanced Scientific Computing Research, Scientific
Discovery through Advanced Computing (SciDAC) program through the
Learning-Accelerated Domain Science (LEADS) SciDAC Institute under
Project Number~SCW1933 at Lawrence Livermore National Laboratory (LLNL)
and through the Center for Hierarchical and Robust Modeling of
Non-Equilibrium Transport (CHaRMNET), a Mathematical Multifaceted
Integrated Capability Center (MMICC), under Award Number~\mbox{DE-SC0023164}.
This work stems from an earlier LLNL Laboratory Directed Research and
Development (LDRD) project (21-FS-042). Lawrence Livermore National
Laboratory is operated by Lawrence Livermore National Security, LLC,
for the United States Department of Energy, National Nuclear Security
Administration under Contract~\mbox{DE-AC52-07NA27344}. LLNL document
release number: LLNL-JRNL-XXXXXXX.

S.~Lee was supported by the National Research Foundation of Korea (NRF)
grant funded by the Korean government's Ministry of Science and
Information and Communication Technology (MSIT), RS-2026-25505849.

\section*{Code, data, and reproducibility}

Code and reproduction instructions are available at
\url{https://github.com/teeratornk/deflation-example}.
The
\href{https://github.com/teeratornk/deflation-example/tree/v0.6.2}
{version 0.6.2 release}
packages both body-fitted meshes, prescribed fields, all four optimization
methods, verification and CPU-only figure regeneration of the earlier body-fitted
campaign. The
\href{https://github.com/teeratornk/deflation-example/blob/ba682dd05b168d2f43ff0c58050813deee785223/docs/prescribed-speedup-v21.md}{skew-transport records}
map every body-fitted table and figure, including the retention ablation and the
Cartesian rows of the scale and trajectory comparisons, to its run commit, evidence
archive and regeneration command; the
\href{https://github.com/teeratornk/deflation-example/tree/mesh-cht-data-v1}{mesh-cht-data-v1 records}
retain the earlier body-fitted campaign. The separate
\href{https://github.com/teeratornk/deflation-example/tree/mesh-cht-refinement-data-v1}{mesh-cht-refinement-data-v1 records}
contain the finer-transformer diagnostics and residual-correction study.
The host-restricted Cartesian comparisons retain their
\href{https://github.com/teeratornk/deflation-example/tree/v0.4.0}{version 0.4.0 numerical implementation}
and
\href{https://github.com/teeratornk/deflation-example/tree/reference-policy-data-v1}
{reference-policy-data-v1 data version}.
The versioned
\href{https://github.com/teeratornk/deflation-example/blob/ba682dd05b168d2f43ff0c58050813deee785223/docs/mesh-showcases.md}{body-fitted}
and \href{https://github.com/teeratornk/deflation-example/blob/v0.4.0/docs/reference-policy-study.md}{Cartesian}
repository guides give configurations and reproduction commands, including the
study commits of the timed body-fitted procedures.
The \href{https://github.com/teeratornk/deflation-example/blob/ba682dd05b168d2f43ff0c58050813deee785223/docs/supplementary-tables.md}
{versioned repository supplement}, \texttt{docs/supplementary-tables.md},
contains the full per-case Cartesian accuracy, rank, timing and memory tables, the
finer-transformer scalar diagnostics and the tables of the earlier body-fitted campaign,
with links to their reproduction examples.
The separate \href{https://github.com/teeratornk/deflation-example/blob/f2daf4aaf05b07453b210ffc3a3d695edefaa266/examples/temperature_bounds/README.md}{temperature-bound example} %
contains the optimized fields at different upper limits, their verification
and the plotting commands.

\section*{Declaration of competing interest}

The authors declare no competing interests.

\section*{Declaration of generative AI and AI-assisted technologies in the manuscript preparation process}

The authors used Claude (Anthropic) and ChatGPT (OpenAI) during
manuscript preparation, primarily for editorial assistance, code
review, and prose polishing; all scientific content, mathematical
derivations, experimental design, data analysis, and conclusions are
the authors' own. 
\FloatBarrier
\putbib[references]
\label{main:lastpage}
\clearpage
\end{bibunit} 
\end{document}